\documentclass[lettersize,journal]{IEEEtran}
\usepackage{amsmath,amsfonts}
\usepackage{algorithmic}
\usepackage{algorithm}
\usepackage{array}
\usepackage[caption=false,font=footnotesize,labelfont=rm,textfont=rm]{subfig}
\usepackage{textcomp}
\usepackage{stfloats}
\usepackage{url}
\usepackage{verbatim}
\usepackage{graphicx}
\usepackage{cite}
\usepackage{hyperref}
\usepackage{tabularray}
\usepackage{color}
\usepackage{booktabs}
\usepackage{amsmath}
\usepackage{amssymb}
\newtheorem{theorem}{Theorem} 
 
\newtheorem{lemma}{Lemma} 

\newtheorem{proposition}{Proposition}
\newtheorem{remark}{Remark}

\begin{document}

\title{Low-Altitude UAV-Assisted Bistatic ISAC: Closed-form 3D CRLB and Coverage Analysis}

\author{Haoyu Jiang,~\IEEEmembership{Student Member,~IEEE,} Hengyou Kong, Xiaoli Xu,~\IEEEmembership{Member,~IEEE,} Yong Zeng,~\IEEEmembership{Fellow,~IEEE,}

\thanks{ 
H. Jiang, H. Kong, X. Xu and Y. Zeng are with the National Mobile Communications Research Laboratory, Southeast University, Nanjing 210096, China. Y. Zeng is also with the Purple Mountain Laboratories, Nanjing 211111, China (e-mail: \{230258936, 230268151, xiaolixu, yong\_zeng\}@seu.edu.cn). (Corresponding author: Yong Zeng.)}

}

\maketitle

\begin{abstract}
This paper investigates the fundamental performance limits of three-dimensional (3D) localization in unmanned aerial vehicle (UAV)-assisted integrated sensing and communication (ISAC) systems. Specifically, a base station (BS) estimates the 3D position of a sensing target with the aid of a UAV acting as a flexible aerial anchor node. We derive a closed-form expression for the 3D Cramér–Rao lower bound (CRLB), which explicitly quantifies the achievable localization accuracy as a function of both the UAV’s location and the target’s position. The CRLB is shown to decompose naturally into three distinct components, arising from signal propagation delay, angular measurements, and their coupling effect, respectively. To validate the analytical results, we consider a representative orthogonal frequency-division multiplexing (OFDM)-based ISAC system and demonstrate that the derived CRLB closely predicts the performance of maximum-likelihood estimation across diverse geometric configurations and UAV mobility patterns. Furthermore, we introduce the notion of CRLB-constrained sensing coverage to characterize the spatial region within which a prescribed localization accuracy can be guaranteed. Through local boundary approximations and coverage-size evaluations, we reveal how UAV displacement, altitude, and the CRLB threshold jointly shape the extent and geometry of the reliable sensing region.
\end{abstract}


\section{Introduction}

\IEEEPARstart{T}{he} sixth-generation (6G) wireless networks are expected to evolve from connection-centric systems toward environment-aware infrastructures that support data transmission, target detection, localization, and situational awareness \cite{liu_jsac_isac_2022,liu_cst_limits_2022,qianglong_tutorial,Zhang_ISAC_tutorial,Wei_ISAC_IoT}. In low-altitude wireless networks, the integration of unmanned aerial vehicles (UAVs) and integrated sensing and communication (ISAC) can be broadly viewed from two paradigms. The first is ISAC for UAVs, where UAVs or UAV swarms are aerial targets to be detected, localized, and tracked by terrestrial or distributed sensing infrastructures \cite{jiang_raa_uav_isac,Mu_UAV_comag,yuxuan_magazine}. The second is UAV-assisted ISAC, where UAVs serve as flexible aerial wireless nodes to actively support both communication and sensing tasks \cite{zeng_procieee_uav_2019,mozaffari_cst_uav_2019,jiang_commmag_lae_isac_2025}. Compared with fixed terrestrial infrastructures, UAV-assisted ISAC can exploit controllable mobility, favorable line-of-sight (LoS) air-ground channels, and adjustable sensing geometry to extend wireless sensing and communication capabilities for infrastructure inspection, traffic monitoring, emergency response, aerial logistics, and other low-altitude economy applications.

In ISAC systems, an important sensing objective is to estimate the state or position of a target from the received wireless signals. To evaluate the fundamental accuracy limit of such inference tasks, Cram\'{e}r-Rao lower bound (CRLB) has been widely adopted as a rigorous performance metric in statistical signal processing, array processing, and wireless sensing \cite{kay_book_1993,stoica_nehorai_tassp_1989,stoica_moses_book_2005}. It characterizes the lowest achievable mean-squared error of unbiased estimators, and thus provides a principled way to assess sensing accuracy beyond received echo power alone. CRLB depends on multiple physical-layer factors, including transmit power, observation time, signal bandwidth, array aperture, propagation distance, and observation direction. As a result, CRLB analysis is useful for understanding how wireless resources and spatial deployment jointly determine the fundamental limit of sensing capability of an ISAC system.

A substantial body of work has investigated CRLB-based sensing performance analysis and ISAC design. For example,
for near-field sensing with extremely large-scale MIMO (XL-MIMO), closed-form CRLBs have been derived under the
spherical-wave propagation model \cite{wang_tsp_nearfield_crlb_2024,Yuan_nearfield_tccn}. With the CRLB as a sensing-accuracy metric, existing ISAC studies have optimized transmit beamforming, covariance matrices, and waveform resources to improve sensing accuracy while maintaining communication performance \cite{liu_tsp_crlb_2022,hua_twc_mimo_2024,ren_twc_fundamental_2024}. Intelligent reflecting surfaces, pinching antennas, sparse MIMO arrays, movable antennas, or multiple coordinated transmitters/receivers have also been exploited to enrich the sensing observations available at the receiver \cite{song_tsp_irs_2023,Jiang_pinching_tcom,Xiu_movable_twc,min_sparse_tsp,cheng_twc_networked_2024,mao_twc_cellfree_2024,xuxiaoli}. In parallel, environment-aware approaches, such as channel knowledge map (CKM)-based design, have been investigated to exploit location-specific prior channel information for improving sensing and communication performance \cite{Wu_ckm_twc}. These works provide important theoretical and algorithmic foundations for CRLB-based ISAC design. However, these existing CRLB results mainly focus on intermediate sensing parameters, such as delay, angle, or Doppler, rather than the eventual  target position error. Even when the position-estimation CRLB can be computed, the impact of the anchor-node position on the localization CRLB is not explicitly revealed in closed form. Besides, to the best of our knowledge, a closed-form CRLB expression for three-dimensional (3D) localization that explicitly characterizes the dependence on the anchor-node position remains unavailable. Such a result is important for guiding CRLB-oriented anchor node positioning, which is especially appealing for UAV-assisted ISAC with UAV serving as flexible aerial anchor node. 

Compared to conventional terrestrial ISAC systems, UAV-assisted ISAC brings an additional degree of freedom, for which the position of the UAV anchor node can be flexibly determined \cite{Fei_UAVisac_commmag,Mao_uavisac_wirecom}. This controllability is qualitatively different from adjusting a transmit beam or allocating waveform resources, because moving the UAV changes the physical observation geometry itself. Such mobility has been extensively exploited for UAV-assisted wireless communications \cite{zeng_trajectory_twc,zeng2026capacity}, and it has been extended to UAV-assisted ISAC via trajectory optimization, maneuver control, beamforming, periodic sensing, and joint communication-localization tasks \cite{lyu_twc_maneuver_2023,deng_twc_adaptable_2023}. These studies confirm that UAV mobility is a valuable design freedom for ISAC. They also suggest a fundamental question: before optimizing a full trajectory or a complete resource-allocation policy, it is informative to understand how a UAV's spatial deployment influences the fundamental limit of  localization accuracy of the system.

Target localization in ISAC is determined not only by the received sensing SNR, but also by how delay and angular observations constrain the target position after being transformed into the Cartesian domain. Since this transformed information depends on the relative geometry among the BS, UAV, target, and receiving array, different UAV placements may lead to substantially different localization accuracy, or even geometrically ill-conditioned sensing configurations. This view is also relevant to active environmental sensing, where the UAV is deployed not only to maintain a wireless link, but also to acquire sensing information about an operating space whose detailed target distribution may not be known beforehand, similar in spirit to coverage-oriented ISAC systems \cite{Gan_ISACcoverage_jsac}.

In this paper, we investigate a UAV-assisted bistatic ISAC system for 3D target localization, where a UAV acts as a flexible aerial anchor, and the target position is estimated from the bistatic delay and the angular information measured at the BS. We derive a closed-form CRLB expression that explicitly characterizes the achievable localization accuracy as a function of both the UAV location and the target position. The derived CRLB further reveals the individual contributions of propagation delay, angular measurements, and their coupling effect. Based on this result, we characterize the CRLB-constrained sensing coverage, namely the spatial region where a prescribed localization-accuracy requirement can be guaranteed, and investigate how UAV deployment, altitude, and the CRLB threshold affect the reliable sensing region.

The main contributions of this paper are summarized as follows.
\begin{itemize}
    \item
    We derive a closed-form position-estimation CRLB for 3D UAV-assisted bistatic ISAC. The derived result is general in the sense that it is not restricted to a specific waveform realization or a fixed UPA size. The expression further provides a geometric interpretation by separating the effects of delay information, angular information, and their coupling in the Cartesian position domain. This interpretation explains how localization accuracy is governed not only by received signal strength, but also by the geometric conditioning of the sensing observations. Based on this CRLB, we reveal how UAV deployment and target position reshape the localization information available at the BS. 

    \item
    We specialize the general CRLB expression to an OFDM-based ISAC waveform, thereby linking the theoretical localization bound with practical ISAC waveform parameters. Based on this specialization, numerical results are provided under different UAV positions and representative UAV movement patterns. The results verify that the derived CRLB correctly predicts the localization-performance variation caused by UAV geometry. More importantly, the validation illustrates that even with the same waveform and resource allocation, different UAV deployment locations can lead to substantially different localization accuracy due to the geometry-dependent Fisher information structure.

    \item
    We analyze the sensing coverage area based on the CRLB criterion to characterize the spatial region where a UAV-assisted ISAC system can support under a prescribed CRLB threshold. The analysis reveals how reliable localization regions are formed around the BS and the UAV, and explains their contour deformation through local closed-form boundary approximations. We further characterize the coverage area in two-dimensional sensing and the coverage volume in 3D sensing, including both total coverage and UAV-side coverage. The results show that the coverage size can vary non-monotonically with UAV displacement and can be significantly affected by the CRLB threshold and UAV altitude, providing insights for flexible UAV anchor node deployment in UAV-assisted ISAC systems.
\end{itemize}

The remainder of this paper is organized as follows. Section~\ref{sec:system_model} presents the UAV-assisted bistatic ISAC system model. 
Section~\ref{sec:crlb_3d} derives the closed-form 3D localization CRLB, provides its geometric decomposition, and discusses the two-dimensional special case. Section~\ref{sec:sensing_coverage} analyzes the CRLB-based sensing coverage and evaluates the deployment-dependent coverage size. 
Section~\ref{sec:verification} validates the closed-form CRLB results using OFDM waveform with maximum-likelihood (ML) estimation and provides numerical results for the sensing coverage analysis. 
Finally, Section~\ref{sec:conclusion} concludes this paper.

\emph{Notations:}
Scalars, vectors, matrices, and sets are denoted by italic, boldface lower-case, boldface upper-case, and calligraphic letters, respectively. For a vector $\mathbf x$, $\|\mathbf x\|$ denotes its Euclidean norm. For a matrix $\mathbf A$, $\mathbf A^{\rm T}$, $\mathbf A^{\rm H}$, $\mathbf A^{-1}$, and ${\rm Tr}(\mathbf A)$ denote its transpose, Hermitian transpose, inverse, and trace, respectively. The symbols $\mathbf I_N$, $\mathbb E\{\cdot\}$, ${\rm Re}\{\cdot\}$, $|\cdot|$, $\otimes$, and $j$ denote the $N$-dimensional identity matrix, expectation, real part, absolute value, Kronecker product, and imaginary unit, respectively.

\section{System Model}\label{sec:system_model}
\begin{figure}[htbp] 
        \centering \includegraphics[width=0.9\columnwidth]{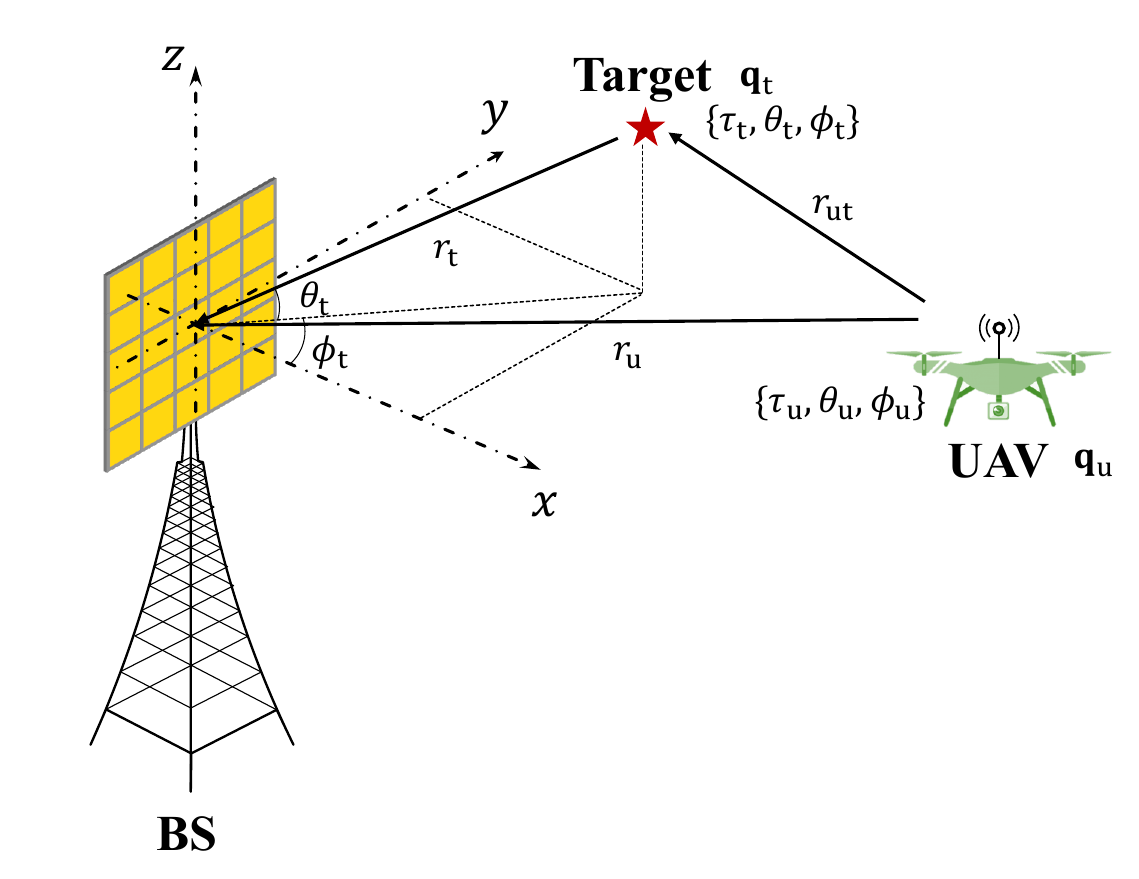}
        \caption{\label{fig:env} An illustration of UAV-assisted uplink bistatic ISAC system.}
\end{figure}
As shown in Fig.~\ref{fig:env}, we consider a UAV-assisted bistatic ISAC system consisting of a BS, a UAV, and a target to be localized. The BS is located at the origin of a 3D Cartesian coordinate system and is equipped with a UPA deployed on the $yz$-plane. The UPA contains $N_y$ and $N_z$ antenna elements along the $y$- and $z$-axes, respectively, and the total number of BS antennas is $N_{\rm R}=N_yN_z$. The inter-element spacing is half wavelength, i.e., $d=\lambda/2$, where $\lambda=c/f_c$, $f_c$ is the carrier frequency, and $c$ is the speed of light. The target and UAV locations are denoted by $\mathbf q_{\rm t}=[x_{\rm t},y_{\rm t},z_{\rm t}]^{\rm T}$ and $\mathbf q_{\rm u}=[x_{\rm u},y_{\rm u},z_{\rm u}]^{\rm T}$, respectively. Their spherical-coordinate representations are given by
\begin{equation}\label{eq:spherical_coordinate}
    \mathbf q_i
    =
    r_i\big[\cos\theta_i\cos\phi_i,\cos\theta_i\sin\phi_i,\sin\theta_i \big]^{\rm T}
    \quad i\in\{{\rm u},{\rm t}\},
\end{equation}
where $r_i=\|\mathbf q_i\|$ denotes the range, $\theta_i$ denotes the elevation angle, and $\phi_i$ denotes the azimuth angle. The UAV-target distance is denoted by $r_{\rm ut}=\|\mathbf q_{\rm u}-\mathbf q_{\rm t}\|$.
Accordingly, the direct UAV-BS delay and the bistatic sensing delay associated with the target are respectively given by $\tau_{\rm u}=r_{\rm u}/c$ and $\tau_{\rm t}=(r_{\rm t}+r_{\rm ut})/c$.

Since the UPA is deployed on the $yz$-plane, the phase progression across the array is determined by the $y$- and $z$-direction components of the incident wave vector. The UPA steering vector associated with direction $(\theta,\phi)$ is given by
\begin{equation}\label{eq:upa_steering}
    \mathbf a(\theta,\phi)
    =
    \mathbf a_y(\theta,\phi)\otimes\mathbf a_z(\theta),
\end{equation}
where
\begin{equation}\label{eq:ay_steering}
    \mathbf a_y(\theta,\phi)
    =
    \left[
    1,
    e^{-j\pi\cos\theta\sin\phi},
    \ldots,
    e^{-j\pi(N_y-1)\cos\theta\sin\phi}
    \right]^{\rm T},
\end{equation}
and
\begin{equation}\label{eq:az_steering}
    \mathbf a_z(\theta)
    =
    \left[
    1,
    e^{-j\pi\sin\theta},
    \ldots,
    e^{-j\pi(N_z-1)\sin\theta}
    \right]^{\rm T}.
\end{equation}
It follows that $\|\mathbf a(\theta,\phi)\|^2=N_{\rm R}$.

The UAV transmits a sensing-capable information-bearing waveform over an observation interval of duration $T_{\rm obs}$. Let $s(t)$ denote the normalized complex baseband waveform satisfying
\begin{equation}\label{eq:waveform_normalization}
    \frac{1}{T_{\rm obs}}
    \int_{\mathcal T_{\rm obs}}
    |s(t)|^2{\rm d}t
    =
    1,
    \quad
    \mathcal T_{\rm obs}\triangleq[0,T_{\rm obs}].
\end{equation}
The transmitted ISAC signal is
\begin{equation}\label{eq:tx_signal}
    x(t)=\sqrt{P_{\rm t}}s(t),
\end{equation}
where $P_{\rm t}$ denotes the average transmit power of the UAV.

Taking the Doppler effect into account, the received signal at the BS can be modeled as
\begin{equation}\label{eq:rx_signal}
\begin{aligned}
    \mathbf y(t)
    &=
    \sqrt{P_{\rm t}}\alpha_{\rm u}
    \mathbf a(\theta_{\rm u},\phi_{\rm u})
    s(t-\tau_{\rm u})
    e^{j2\pi\nu_{\rm u}t}
    \\
    &\quad+
    \sqrt{P_{\rm t}}\alpha_{\rm t}
    \mathbf a(\theta_{\rm t},\phi_{\rm t})
    s(t-\tau_{\rm t})
    e^{j2\pi\nu_{\rm t}t}
    +
    \mathbf n(t),
    \quad t\in\mathcal T_{\rm obs},
\end{aligned}
\end{equation}
where $\alpha_{\rm u}$ and $\alpha_{\rm t}$ denote the complex channel gains of the direct UAV-BS path and the target-reflected path, respectively, $\nu_{\rm u}$ and $\nu_{\rm t}$ denote the corresponding Doppler shifts, and $\mathbf n(t)$ is complex additive white Gaussian noise with complex-baseband power spectral density $N_0$. In the subsequent CRLB analysis, the focus is on the geometry-dependent angle-delay information for target localization, where the Doppler shift is assumed to be compensated. Under the Friis propagation model, the direct-path channel power gain is modeled as
\begin{equation}
    |\alpha_{\rm u}|^2
    =
    \zeta_0r_{\rm u}^{-2},
    \quad
    \zeta_0\triangleq
    \frac{c^2}{(4\pi)^2f_c^2}.
\end{equation}
For the target-reflected path, the bistatic radar equation gives
\begin{equation}
    |\alpha_{\rm t}|^2
    =
    \xi_0 r_{\rm ut}^{-2}r_{\rm t}^{-2},
    \quad
    \xi_0\triangleq
    \frac{c^2\kappa}{(4\pi)^3f_c^2},
\end{equation}
where $\kappa$ denotes the target radar cross section (RCS). 


\section{Closed-form CRLB for 3D Localization }\label{sec:crlb_3d}

In this section, we characterize how the UAV deployment location $\mathbf{q}_{\rm u}$ affects the 3D sensing accuracy for the target position $\mathbf{q}_{\rm t}$. To this end, we first derive the Fisher information associated with the angle-delay parameters and then transform it into the Cartesian position domain. The resulting CRLB is expressed as a closed-form function of the target location $\mathbf{q}_{\rm t}$ and the UAV location $\mathbf{q}_{\rm u}$. This closed-form expression enables us to identify the geometric factors that govern the 3D localization accuracy and to further evaluate CRLB-based sensing coverage.

\subsection{SNR and Fisher Information}
For the sensing functionality of the considered ISAC system, the BS estimates the target location from the target-reflected path. The direct UAV-BS component does not contain target-location information and is assumed to be suppressed or calibrated before sensing processing. After Doppler compensation, the target-reflected sensing signal can be written as
\begin{equation}\label{eq:sensing_signal}
    \mathbf y_{\rm s}(t)
    =
    \sqrt{P_{\rm t}}\alpha_{\rm t}
    \mathbf a(\theta_{\rm t},\phi_{\rm t})
    s(t-\tau_{\rm t})
    +
    \mathbf n(t),
    \quad t\in\mathcal T_{\rm obs}.
\end{equation}
Since the target parameters are estimated from the observations collected within $T_{\rm obs}$, the accumulated sensing observation SNR is defined as \cite{van1968detection}
\begin{equation} 
\begin{aligned} 
\Upsilon_{\rm s} &= \frac{|\alpha_{\rm t}|^2 P_{\rm t}T_{\rm obs} \|\mathbf a(\theta_{\rm t},\phi_{\rm t})\|^2}{N_0} \\ &= \frac{P_{\rm t}T_{\rm obs}N_{\rm R}\xi_0} {N_0 r_{\rm ut}^2r_{\rm t}^2} = \xi_1 r_{\rm ut}^{-2}r_{\rm t}^{-2}, 
\end{aligned} \label{eq:sensing_snr} 
\end{equation} 
where $\xi_1\triangleq P_{\rm t}T_{\rm obs}N_{\rm R}\xi_0/N_0$ is the reference SNR.

We first consider the Fisher information associated with the angle-delay parameter vector
\begin{equation}\label{eq:theta_vector}
    \boldsymbol\Theta
    =
    [\theta_{\rm t},\phi_{\rm t},\tau_{\rm t}]^{\rm T}.
\end{equation}
The equivalent Fisher information for the bistatic delay $\tau_{\rm t}$ is \cite{liu_cst_limits_2022}
\begin{equation}\label{eq:delay_fim}
    I_{\tau_{\rm t}}
    =
    8\pi^2\beta_{\rm rms}^2\Upsilon_{\rm s}
    =
    \frac{\xi_\tau}{r_{\rm ut}^2r_{\rm t}^2},
\end{equation}
where 
\begin{equation}\label{eq:rms_bandwidth}
    \beta_{\rm rms}^2
    =
    \frac{\int (f-\bar f)^2|S(f)|^2{\rm d}f}
    {\int |S(f)|^2{\rm d}f},
    \quad
    \bar f
    =
    \frac{\int f|S(f)|^2{\rm d}f}
    {\int |S(f)|^2{\rm d}f},
\end{equation}
$S(f)$ is the Fourier transform of $s(t)$, and $\xi_\tau\triangleq8\pi^2\beta_{\rm rms}^2\xi_1$.


\begin{lemma}\label{lem:angular_fim}
For the UPA deployed on the $yz$-plane, the effective angular Fisher information matrix for $(\theta_{\rm t},\phi_{\rm t})$ is
\begin{equation}\label{eq:angular_fim}
    \mathbf I_{\rm ang}
    =
    \begin{bmatrix}
        I_{\theta_{\rm t}\theta_{\rm t}} & I_{\theta_{\rm t}\phi_{\rm t}}\\
        I_{\theta_{\rm t}\phi_{\rm t}} & I_{\phi_{\rm t}\phi_{\rm t}}
    \end{bmatrix},
\end{equation}
where
\begin{equation}\label{eq:I_theta_theta}
    I_{\theta_{\rm t}\theta_{\rm t}}
    =
    \frac{
    \xi_y\sin^2\theta_{\rm t}\sin^2\phi_{\rm t}
    +
    \xi_z\cos^2\theta_{\rm t}
    }
    {r_{\rm ut}^2r_{\rm t}^2},
\end{equation}
\begin{equation}\label{eq:I_phi_phi}
    I_{\phi_{\rm t}\phi_{\rm t}}
    =
    \frac{
    \xi_y\cos^2\theta_{\rm t}\cos^2\phi_{\rm t}
    }
    {r_{\rm ut}^2r_{\rm t}^2},
\end{equation}
and
\begin{equation}\label{eq:I_theta_phi}
    I_{\theta_{\rm t}\phi_{\rm t}}
    =
    -
    \frac{
    \xi_y\sin\theta_{\rm t}\cos\theta_{\rm t}
    \sin\phi_{\rm t}\cos\phi_{\rm t}
    }
    {r_{\rm ut}^2r_{\rm t}^2},
\end{equation}
with
\begin{equation}\label{eq:xi_y_z}
    \xi_y\triangleq
    \frac{\pi^2(N_y^2-1)\xi_1}{6},
    \quad
    \xi_z\triangleq
    \frac{\pi^2(N_z^2-1)\xi_1}{6}.
\end{equation}
\end{lemma}

\begin{IEEEproof}
Please refer to Appendix~\ref{app:angle_information}.
\end{IEEEproof}

Thus, the Fisher information matrix for $\boldsymbol\Theta=[\theta_{\rm t},\phi_{\rm t},\tau_{\rm t}]^{\rm T}$ is
\begin{equation}\label{eq:theta_fim}
    \mathbf I(\boldsymbol\Theta)
    =
    \begin{bmatrix}
        I_{\theta_{\rm t}\theta_{\rm t}} & I_{\theta_{\rm t}\phi_{\rm t}} & 0\\
        I_{\theta_{\rm t}\phi_{\rm t}} & I_{\phi_{\rm t}\phi_{\rm t}} & 0\\
        0 & 0 & I_{\tau_{\rm t}}
    \end{bmatrix}.
\end{equation}

\subsection{CRLB for 3D Localization}\label{subsec:3dCRLB}
Since the sensing objective of the ISAC system is 3D target localization, the Fisher information should be transformed from the angle-delay domain to the Cartesian position domain. The transformation from the target location $\mathbf q_{\rm t}$ to the angle-delay parameters is
\begin{equation}\label{eq:parameter_transformation}
    \boldsymbol\Theta
    =
    \mathbf g(\mathbf q_{\rm t},\mathbf q_{\rm u})
    =
    \begin{bmatrix}
        \arcsin(z_{\rm t}/r_{\rm t})\\
        \arctan(y_{\rm t}/x_{\rm t})\\
        (r_{\rm t}+r_{\rm ut})/c
    \end{bmatrix}.
\end{equation}
Define
\begin{equation}\label{eq:local_basis_1}
    \mathbf e_r
    =
    \begin{bmatrix}
        \cos\theta_{\rm t}\cos\phi_{\rm t}\\
        \cos\theta_{\rm t}\sin\phi_{\rm t}\\
        \sin\theta_{\rm t}
    \end{bmatrix},
    \quad
    \mathbf e_\theta
    =
    \begin{bmatrix}
        -\sin\theta_{\rm t}\cos\phi_{\rm t}\\
        -\sin\theta_{\rm t}\sin\phi_{\rm t}\\
        \cos\theta_{\rm t}
    \end{bmatrix},
\end{equation}
and
\begin{equation}\label{eq:local_basis_2}
    \mathbf e_\phi
    =
    \begin{bmatrix}
        -\sin\phi_{\rm t}\\
        \cos\phi_{\rm t}\\
        0
    \end{bmatrix},
    \quad
    \mathbf e_{\rm ut}
    =
    \frac{\mathbf q_{\rm t}-\mathbf q_{\rm u}}{r_{\rm ut}},
\end{equation}
where
\begin{equation}
\begin{aligned}
    r_{\rm ut}=\Big(
&r_{\rm u}^{2}+r_{\rm t}^{2}
-2r_{\rm u}r_{\rm t}
\big[
\sin\theta_{\rm u}\sin\theta_{\rm t} \\
&+
\cos\theta_{\rm u}\cos\theta_{\rm t}
\cos(\phi_{\rm u}-\phi_{\rm t})
\big]
\Big)^{\frac{1}{2}}.
\end{aligned}
\end{equation}
The corresponding Jacobian matrix is
\begin{equation}\label{eq:jacobian_qt}
    \mathbf J_{\mathbf q_{\rm t}}
    =
    \frac{\partial\boldsymbol\Theta}{\partial\mathbf q_{\rm t}}
    =
    \begin{bmatrix}
        r_{\rm t}^{-1}\mathbf e_\theta^{\rm T}\\
        (r_{\rm t}\cos\theta_{\rm t})^{-1}\mathbf e_\phi^{\rm T}\\
        c^{-1}(\mathbf e_r+\mathbf e_{\rm ut})^{\rm T}
    \end{bmatrix}.
\end{equation}
Therefore, the equivalent Fisher information matrix for target localization is
\begin{equation}\label{eq:position_fim}
    \mathbf F(\mathbf q_{\rm t},\mathbf q_{\rm u})
    =
    \mathbf J_{\mathbf q_{\rm t}}^{\rm T}
    \mathbf I(\boldsymbol\Theta)
    \mathbf J_{\mathbf q_{\rm t}}.
\end{equation}
The position-domain CRLB is defined as
\begin{equation}\label{eq:position_crlb_def}
    \mathcal C(\mathbf q_{\rm t},\mathbf q_{\rm u})
    =
    {\rm Tr}
    \left(
    \mathbf F^{-1}(\mathbf q_{\rm t},\mathbf q_{\rm u})
    \right),
\end{equation}
which lower-bounds the mean-squared localization error as
\begin{equation}\label{eq:mse_crlb}
    \mathbb E
    \left[
    \|\hat{\mathbf q}_{\rm t}-\mathbf q_{\rm t}\|^2
    \right]
    \ge
    \mathcal C(\mathbf q_{\rm t},\mathbf q_{\rm u}),
\end{equation}
where $\hat{\mathbf q}_{\rm t}$ is the estimated target location.

\begin{theorem}\label{the:3d_crlb}
For the considered 3D UAV-assisted bistatic ISAC system, the CRLB for target position estimation is given by
\begin{equation}\label{eq:3d_crlb}
\begin{aligned}
    &\mathcal C(\mathbf q_{\rm t},\mathbf q_{\rm u})
    =
    \frac{
    4c^2r_{\rm ut}^4r_{\rm t}^4
    }
    {
    \xi_\tau
    \left((r_{\rm t}+r_{\rm ut})^2-r_{\rm u}^2\right)^2
    }
    \\
    &\quad+
    \frac{
    r_{\rm ut}^2r_{\rm t}^4
    }
    {
    \xi_y\cos^2\phi_{\rm t}
    }
    \left(
    1+
    \frac{
    4r_{\rm u}^2r_{\rm t}^2
    \cos^2\theta_{\rm u}
    \sin^2(\phi_{\rm u}-\phi_{\rm t})
    }
    {
    \left((r_{\rm t}+r_{\rm ut})^2-r_{\rm u}^2\right)^2
    }
    \right)
    \\
    &\quad+
    \frac{
    r_{\rm ut}^2r_{\rm t}^4
    }
    {
    \xi_z\cos^2\theta_{\rm t}\cos^2\phi_{\rm t}
    }
    \Bigg(
    1-\cos^2\theta_{\rm t}\sin^2\phi_{\rm t}
    \\
    &\qquad\
    +
    \frac{
    4r_{\rm u}^2r_{\rm t}^2
    (
    \cos\theta_{\rm u}\sin\theta_{\rm t}\cos\phi_{\rm u}
    -
    \sin\theta_{\rm u}\cos\theta_{\rm t}\cos\phi_{\rm t}
    )^2
    }
    {
    \left((r_{\rm t}+r_{\rm ut})^2-r_{\rm u}^2\right)^2
    }
    \Bigg).
\end{aligned}
\end{equation}
\end{theorem}

\begin{IEEEproof}
Please refer to Appendix~\ref{app:3d_crlb}.
\end{IEEEproof}

\begin{figure}[t] 
        \centering \includegraphics[width=0.9\columnwidth]{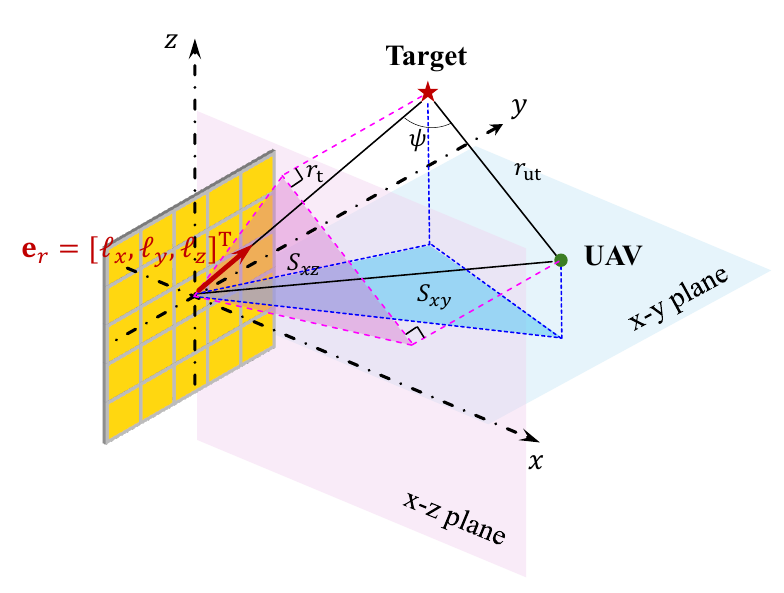}
        \caption{\label{fig:Geometric_Interpretation} An illustration of geometric interpretation of the CRLB.}
\end{figure}

To reveal the geometric meaning of the CRLB in Theorem~\ref{the:3d_crlb}, let $\ell_x$, $\ell_y$, and $\ell_z$ denote the projection of the BS-target unit vector onto the three axes, i.e.,
\begin{equation}\label{eq:direction_cosines}
    [\ell_x,\ell_y,\ell_z]^{\rm T}
    \triangleq
    \mathbf e_r
    =
    \frac{\mathbf q_{\rm t}}{r_{\rm t}}.
\end{equation}
According to \eqref{eq:spherical_coordinate}, they are explicitly given by
\begin{equation}\label{eq:direction_cosines_spherical}
    \ell_x=\cos\theta_{\rm t}\cos\phi_{\rm t},
    \quad
    \ell_y=\cos\theta_{\rm t}\sin\phi_{\rm t},
    \quad
    \ell_z=\sin\theta_{\rm t}.
\end{equation}
Since the BS UPA is deployed on the $yz$-plane, $\ell_x$ is the direction cosine along the normal vector of the UPA. Furthermore, let $\psi$ denote the angle between the BS-target direction $\mathbf e_r=\mathbf q_{\rm t}/r_{\rm t}$ and the UAV-target direction $\mathbf e_{\rm ut}=(\mathbf q_{\rm t}-\mathbf q_{\rm u})/r_{\rm ut}$, and define
\begin{equation}\label{eq:h_definition}
    h
    \triangleq
    1+\cos\psi
    =
    1+\mathbf e_r^{\rm T}\mathbf e_{\rm ut}
    =
    \frac{(r_{\rm t}+r_{\rm ut})^2-r_{\rm u}^2}
    {2r_{\rm t}r_{\rm ut}}.
\end{equation}
Let $S_{xy}$ and $S_{xz}$ denote the areas of the triangles formed by the BS, the target $\mathbf{q}_{\rm t}$, and the UAV $\mathbf{q}_{\rm u}$ after projecting onto the $xy$- and $xz$-planes, respectively, i.e.,
\begin{equation}\label{eq:projected_area_def}
    S_{xy}
    \triangleq
    \frac{1}{2}
    \left|
    \mathbf e_z^{\rm T}
    (\mathbf q_{\rm t}\times\mathbf q_{\rm u})
    \right|,
    \quad
    S_{xz}
    \triangleq
    \frac{1}{2}
    \left|
    \mathbf e_y^{\rm T}
    (\mathbf q_{\rm t}\times\mathbf q_{\rm u})
    \right|,
\end{equation}
where $\mathbf e_y=[0,1,0]^{\rm T}$ and $\mathbf e_z=[0,0,1]^{\rm T}$. From the spherical-coordinate representation, we have
\begin{equation}\label{eq:Sxy}
    4S_{xy}^2
    =
    r_{\rm u}^2r_{\rm t}^2
    \cos^2\theta_{\rm u}\cos^2\theta_{\rm t}
    \sin^2(\phi_{\rm u}-\phi_{\rm t}),
\end{equation}
and
\begin{equation}\label{eq:Sxz}
\begin{aligned}
    4S_{xz}^2
    &=
    r_{\rm u}^2r_{\rm t}^2
    (
    \cos\theta_{\rm u}\sin\theta_{\rm t}\cos\phi_{\rm u}
    \\
    &\quad
    -
    \sin\theta_{\rm u}\cos\theta_{\rm t}\cos\phi_{\rm t}
    )^2.
\end{aligned}
\end{equation}
Using the above direction-cosine and projected-area definitions together with \eqref{eq:h_definition}, the position-domain CRLB in \eqref{eq:3d_crlb} can be equivalently written as
\begin{equation}\label{eq:crlb_geo_form}
\begin{aligned}
    \mathcal C(\mathbf q_{\rm t},\mathbf q_{\rm u})
    &=
    \underbrace{
    \frac{
    c^2r_{\rm ut}^2r_{\rm t}^2
    }
    {
    \xi_\tau h^2
    }
    }_{\mathcal C_\tau}
    +
    \underbrace{
    \frac{
    r_{\rm ut}^2r_{\rm t}^4
    }
    {\ell_x^2}
    \left(
    \frac{1-\ell_z^2}{\xi_y}
    +
    \frac{1-\ell_y^2}{\xi_z}
    \right)
    }_{\mathcal C_{\rm ang}}
    \\
    &\quad+
    \underbrace{
    \frac{
    4r_{\rm t}^2
    }
    {
    h^2\ell_x^2
    }
    \left(
    \frac{S_{xy}^2}{\xi_y}
    +
    \frac{S_{xz}^2}{\xi_z}
    \right)
    }_{\mathcal C_{\rm cpl}}.
\end{aligned}
\end{equation}
In \eqref{eq:crlb_geo_form}, $\mathcal C_\tau$ represents the contribution associated with bistatic delay information, where $h$ characterizes the delay observability determined by the bistatic geometry. The term $\mathcal C_{\rm ang}$ represents the basic contribution associated with the two-dimensional angular information provided by the UPA, where $\ell_x$ characterizes the target direction relative to the array normal. The term $\mathcal C_{\rm cpl}$ represents the geometric coupling between the bistatic delay and angular information, with $S_{xy}$ and $S_{xz}$ measuring the projected bistatic geometry along the two array dimensions. The parameters $\xi_\tau$, $\xi_y$, and $\xi_z$ characterize the effective delay information and the angular information provided by the $y$- and $z$-direction array apertures, respectively.

To gain some insights and further expose the physical meaning of \eqref{eq:crlb_geo_form}, define the dimensionless root mean square (RMS) bandwidth and the dimensionless projected geometric factors as
\begin{equation}
    \bar{\beta}\triangleq\frac{\beta_{\rm rms}}{f_c},
    \qquad
    \bar d_{xy}\triangleq\frac{2S_{xy}}{r_{\rm ut}r_{\rm t}},
    \qquad
    \bar d_{xz}\triangleq\frac{2S_{xz}}{r_{\rm ut}r_{\rm t}}.
    \label{eq:normalized_bandwidth_projected_lengths}
\end{equation}
Here, $\bar{\beta}$ is the normalized RMS bandwidth and is typically much smaller than one in wireless systems, while $\bar d_{xy}$ and $\bar d_{xz}$ are normalized projected geometric factors satisfying $0\leq\bar d_{xy},\bar d_{xz}\leq1$. The CRLB in \eqref{eq:crlb_geo_form} can then be rewritten as
\begin{equation}
\label{eq:crlb_geo_physical_form}
\begin{aligned}
\mathcal C(\mathbf q_{\rm t},\mathbf q_{\rm u})
=
\frac{1}{\Upsilon_{\rm s}}
\Bigg[
&
\frac{\lambda^2}{8\pi^2\bar{\beta}^2h^2}
+
\frac{6r_{\rm t}^2}{\pi^2\ell_x^2}
\left(
\frac{1-\ell_z^2}{N_y^2-1}
+
\frac{1-\ell_y^2}{N_z^2-1}
\right)
\\
&+
\frac{6r_{\rm t}^2}{\pi^2h^2\ell_x^2}
\left(
\frac{\bar d_{xy}^2}{N_y^2-1}
+
\frac{\bar d_{xz}^2}{N_z^2-1}
\right)
\Bigg].
\end{aligned}
\end{equation}
This form shows that the localization CRLB scales inversely with the accumulated sensing SNR $\Upsilon_{\rm s}$. The delay-related term is governed by the wavelength $\lambda$, the normalized bandwidth $\bar{\beta}$, and the bistatic observability factor $h$. The angular and coupling terms scale with $r_{\rm t}^2$ and are further shaped by the UPA aperture sizes, the array visibility factor $\ell_x$, and the dimensionless projected geometric factors $\bar d_{xy}$ and $\bar d_{xz}$.


\begin{remark}[Geometric Singularities]\label{rem:geo_singularities}
From \eqref{eq:crlb_geo_form}, the 3D position CRLB becomes infinite under either of the following two geometric conditions.

\begin{enumerate}
    \item First, the delay-related CRLB term becomes infinite when
        \begin{equation}\label{eq:delay_singularity}
            h=0
            \quad\Longleftrightarrow\quad
            (r_{\rm t}+r_{\rm ut})^2-r_{\rm u}^2=0.
        \end{equation}
        Since $h=0$ is equivalent to $\psi=\pi$, this condition occurs when the target lies on the line segment connecting the BS and the UAV, i.e.,
        \begin{equation}\label{eq:bistatic_collinearity}
            \mathbf q_{\rm u}
            =
            \alpha\mathbf q_{\rm t},
            \quad
            \alpha>1.
        \end{equation}
        In this case, the direction ray may overlap with the constant-bistatic-delay ellipse along the BS-UAV line, producing non-unique intersections, as shown in Fig. \ref{fig:singularity}.

    \item Second, the the angular-information-related CRLB term becomes infinite when
        \begin{equation}\label{eq:angular_singularity}
            \ell_x
            =
            \cos\theta_{\rm t}\cos\phi_{\rm t}
            =
            \frac{x_{\rm t}}{r_{\rm t}}
            =
            0.
        \end{equation}
        Geometrically, this occurs when the target direction lies in the $yz$-plane of the BS UPA, where the array cannot provide any angular information for 3D position estimation. 
\end{enumerate}

Consequently, a finite CRLB requires the target direction to avoid the UPA angular-singularity set and the bistatic geometry to avoid collinearity among the BS, the target, and the UAV.
\end{remark}

\begin{figure}[t]
    \centering
    \subfloat[Unique intersection.]{
        \includegraphics[width=0.44\columnwidth]{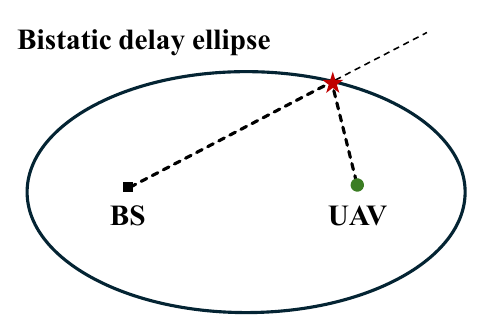}
        \label{fig:singularity1}}
    \hfill
    \subfloat[Non-unique intersections.]{
        \includegraphics[width=0.44\columnwidth]{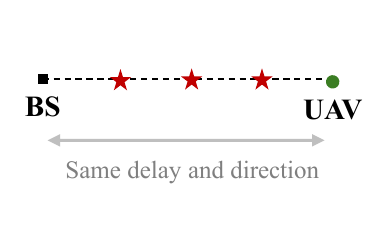}
        \label{fig:singularity2}}
    \caption{Geometric interpretation of the collinear case where CRLB becomes infinite since the target cannot be uniquely localized.}
    \label{fig:singularity}
\end{figure}

\subsection{Two-Dimensional Special Case}
\label{subsec:2d_crlb}

The general 3D CRLB in Theorem~\ref{the:3d_crlb} characterizes the localization performance when both elevation and azimuth information are available from the BS UPA. When the UPA reduces to a ULA, i.e., $N_z=1$, by neglecting the altitude difference among the target, UAV and BS, the target-localization problem becomes a two-dimensional (2D) special case, as considered in \cite{xuxiaoli}. Consider the case where both the target and the UAV are restricted to the $xy$-plane, i.e., $\theta_{\rm t}=\theta_{\rm u}=0$. The unknown target position is then reduced from the 3D vector $\mathbf q_{\rm t}=[x_{\rm t},y_{\rm t},z_{\rm t}]^{\rm T}$ to the planar vector $\mathbf p_{\rm t}=[x_{\rm t},y_{\rm t}]^{\rm T}$, with
\begin{equation}
    \mathbf q_{\rm t}
    =
    \mathbf E\mathbf p_{\rm t},
    \qquad
    \mathbf E
    =
    \begin{bmatrix}
    1&0\\
    0&1\\
    0&0
    \end{bmatrix}.
    \label{eq:2d_embedding_matrix}
\end{equation}
Therefore, the 2D CRLB is a reduced-parameter CRLB under the planar constraint, rather than the full three-dimensional CRLB evaluated at $z_{\rm t}=0$.

By the chain rule, the reduced Jacobian with respect to the planar position $\mathbf p_{\rm t}$ is
\begin{equation}
    \mathbf J_{\mathbf p_{\rm t}}
    =
    \frac{\partial\boldsymbol\Theta}{\partial\mathbf p_{\rm t}}
    =
    \frac{\partial\boldsymbol\Theta}{\partial\mathbf q_{\rm t}}
    \frac{\partial\mathbf q_{\rm t}}{\partial\mathbf p_{\rm t}}
    = \left.
    \mathbf J_{\mathbf q_{\rm t}}\mathbf E \right|_{\theta_{\rm t}=0,\ \theta_{\rm u}=0}.
    \label{eq:2d_reduced_jacobian}
\end{equation}
The corresponding two-dimensional Fisher information matrix is thus
\begin{equation}
    \mathbf F_{\rm 2D}
    =
    \mathbf J_{\mathbf p_{\rm t}}^{\rm T}
    \mathbf I(\boldsymbol\Theta)
    \mathbf J_{\mathbf p_{\rm t}}
    =
    \mathbf E^{\rm T}
    \mathbf F(\mathbf q_{\rm t},\mathbf q_{\rm u})
    \mathbf E |_{\theta_{\rm t}=0,\ \theta_{\rm u}=0}.
    \label{eq:2d_fim_reduced}
\end{equation}

\begin{figure}[t] 
        \centering \includegraphics[width=0.7\columnwidth]{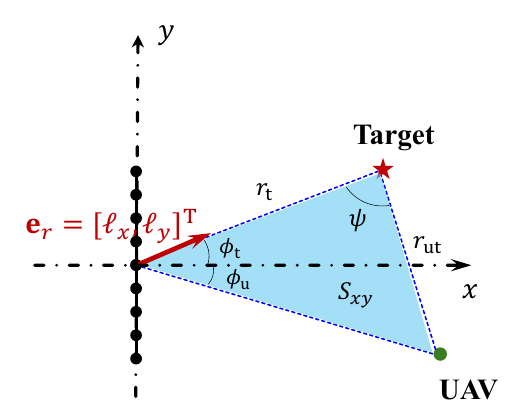}
        \caption{\label{fig:Geometric_Interpretation_2D} An illustration of 2D geometric interpretation of the CRLB.}
\end{figure}

\begin{theorem}\label{the:2d_crlb}
For the 2D UAV-assisted bistatic ISAC system with the BS array reduced to a ULA along the $y$-axis, the CRLB for target position estimation is given by
\begin{equation}\label{eq:2d_crlb}
\begin{aligned}
    &\mathcal C_{\rm 2D}(r_{\rm u},\phi_{\rm u};r_{\rm t},\phi_{\rm t})
    \\
    &=
    \frac{
    4r_{\rm ut}^3r_{\rm t}^5
    }
    {
    \xi_y\cos^2\phi_{\rm t}
    \left((r_{\rm ut}+r_{\rm t})^2-r_{\rm u}^2\right)
    }
    +
    \frac{
    4c^2r_{\rm ut}^4r_{\rm t}^4
    }
    {
    \xi_\tau
    \left((r_{\rm ut}+r_{\rm t})^2-r_{\rm u}^2\right)^2
    } .
\end{aligned}
\end{equation}
\end{theorem}

\begin{IEEEproof}
Substitute \eqref{eq:2d_fim_reduced} into \eqref{eq:position_crlb_def} and expand.
\end{IEEEproof}

\begin{remark}[Geometric Interpretation of the 2D CRLB]\label{rem:2d_geo_interpretation}
Using the symbols defined in \ref{subsec:3dCRLB}, the two-dimensional CRLB in \eqref{eq:2d_crlb} can be equivalently written as
\begin{equation}\label{eq:2d_crlb_geo_form}
\begin{aligned}
    \mathcal C_{\rm 2D}(\mathbf q_{\rm t},\mathbf q_{\rm u})
    &=
    \underbrace{
    \frac{
    c^2r_{\rm ut}^2r_{\rm t}^2
    }
    {
    \xi_\tau h^2
    }
    }_{\mathcal C_\tau^{\rm 2D}}
    +
    \underbrace{
    \frac{
    r_{\rm ut}^2r_{\rm t}^4
    }
    {
    \xi_y\ell_x^2
    }
    }_{\mathcal C_{\rm ang}^{\rm 2D}}
    +
    \underbrace{
    \frac{
    4r_{\rm t}^2S_{xy}^2
    }
    {
    \xi_y h^2\ell_x^2
    }
    }_{\mathcal C_{\rm cpl}^{\rm 2D}} .
\end{aligned}
\end{equation}

\end{remark}


\begin{remark}[Geometric Singularities in the 2D Case]\label{rem:2d_geo_singularities}
From \eqref{eq:2d_crlb_geo_form}, the 2D position CRLB becomes infinite under either of the following two geometric conditions.

\begin{enumerate}
    \item The collinearity among the BS, the target, and the UAV
    \begin{equation}\label{eq:2d_delay_singularity}
        h=0
        \Longleftrightarrow
        \mathbf q_{\rm u}=\alpha\mathbf q_{\rm t},\ \alpha>1 .
    \end{equation}

    \item The BS-target direction is parallel to the array axis
    \begin{equation}\label{eq:2d_angular_singularity}
        \ell_x = \cos\phi_{\rm t} = 0 .
    \end{equation}
\end{enumerate}

\end{remark}

\subsection{CRLB Heatmap and Volume Visualization}\label{subsec:visualization}
To illustrate the proposed CRLB characterization, we provide representative numerical examples in both 2D and 3D settings. The BS is located at the origin, and the visible half-space $x>0$ is considered to avoid the front-back ambiguity of the BS UPA. The main parameters are summarized in Table~\ref{tab:sim_params}. Unless otherwise specified, CRLB clipping is used only for visualization in heatmaps and isosurfaces.

\begin{table}[!t]
\centering
\caption{Simulation Parameters}
\label{tab:sim_params}
\renewcommand{\arraystretch}{1.12}
\begin{tabular}{p{0.43\linewidth}p{0.18\linewidth}p{0.22\linewidth}}
\hline
\textbf{Parameter} & \textbf{Symbol} & \textbf{Value} \\
\hline
Speed of light & $c$ & $3\times 10^8$ m/s \\
Carrier frequency & $f_{\rm c}$ & $24$ GHz \\
UAV transmit power & $P_{\rm t}$ & $1$ W \\
Noise power spectral density & $N_0$ & $-170$ dBm/Hz \\
RCS-related coefficient & $\kappa$ & $0.1$ m$^2$ \\
Observation duration & $T_{\rm obs}$ & $1$ ms \\
RMS bandwidth & $\beta_{\rm rms}$ & $20$ MHz \\
UPA size at the BS & $(N_y,N_z)$ & $(8,8)$ \\
\hline
\end{tabular}
\end{table}

We first visualize the spatial variation of the CRLB. In Fig.~\ref{fig:sim_2d_crlb_target}, the UAV is fixed and the target location varies over the visible half-plane. The CRLB distribution is highly nonuniform. A high-CRLB strip appears along the BS-UAV direction, reflecting the degradation of bistatic observability near the collinearity direction. Around the UAV, the contours are indented toward the BS, showing that the UAV-near sensing geometry is direction-dependent rather than radially symmetric. In Fig.~\ref{fig:sim_2d_crlb_uav}, the target is fixed and the UAV location varies. Since the BS-array angular visibility of the target is fixed, the CRLB variation is mainly induced by the UAV-dependent bistatic geometry and UAV-target distance. The resulting cardioid-like pattern, with a notch along the BS-target direction, indicates that UAV placements with comparable target distances can still yield different localization accuracy.

\begin{figure}[t]
    \centering
    \subfloat[2D CRLB with fixed UAV.]{
        \includegraphics[width=0.45\columnwidth]{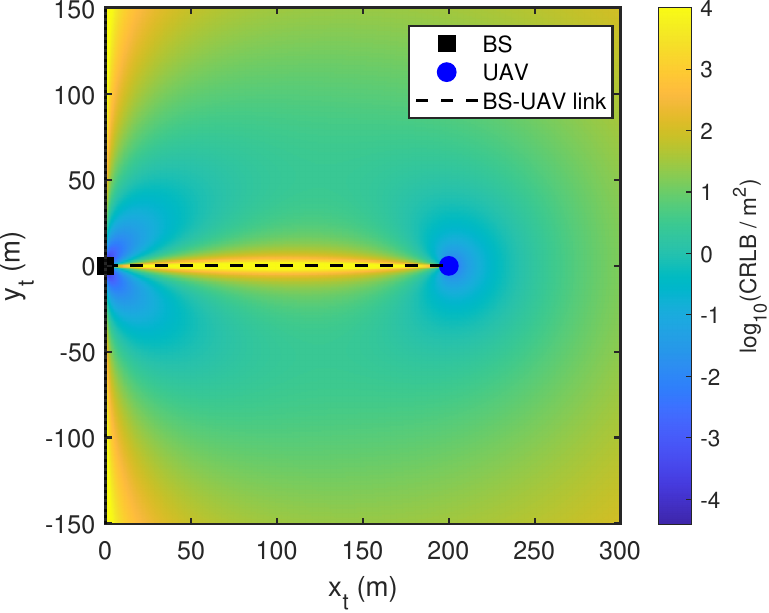}
        \label{fig:sim_2d_crlb_target}}
    \hfill
    \subfloat[2D CRLB with fixed target.]{
        \includegraphics[width=0.45\columnwidth]{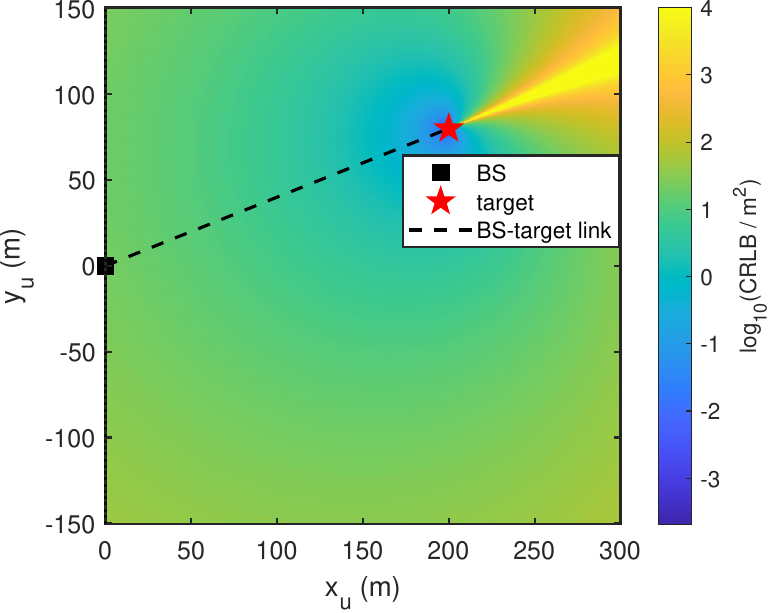}
        \label{fig:sim_2d_crlb_uav}}
    \caption{Plot of $\log_{10}\mathcal C_{\rm 2D}(\mathbf q_{\rm t},\mathbf q_{\rm u})$, by varying the target location $\mathbf q_{\rm t}$, with $\mathbf q_{\rm u}$ fixed to $(100,0)$; and by varying the UAV location $\mathbf q_{\rm u}$, with $\mathbf q_{\rm t}$ fixed to $(100,40)$.}
    \label{fig:sim_2d_crlb_heatmap}
\end{figure}

Fig.~\ref{fig:sim_3d_crlb_volume} further illustrates the 3D CRLB distribution. 
In Fig.~\ref{fig:sim_3d_crlb_target}, the UAV is fixed and the target location varies. The isosurfaces are clearly non-spherical. The BS-near region forms a dimpled-cap-shaped body whose indentation is aligned with the BS-UAV link, while the UAV-near region exhibits a cusp-like indentation toward the BS. For a looser CRLB level, these two regions may merge into a connected body with an internal cavity. In Fig.~\ref{fig:sim_3d_crlb_uav}, the target is fixed and the UAV location is scanned. The isosurface is centered around the target region but remains anisotropic. Since the target direction relative to the BS UPA is fixed, $\ell_x$ remains unchanged, and the shape variation is mainly governed by the UAV-dependent factors $h$, $S_{xy}$, and $S_{xz}$. This confirms that UAV placement affects the CRLB through propagation loss, bistatic delay observability, and angle-delay coupling.

This volume visualization can be interpreted as a CRLB-based counterpart of the Cassini oval in bistatic radar. For a prescribed CRLB level $\Gamma$, we define the \textit{Cram\'{e}r-Rao surface} as
\begin{equation}\label{eq:crlb_surface}
    \mathcal S_{\rm C}(\Gamma;\mathbf q_{\rm u})
    \triangleq
    \left\{
    \mathbf q_{\rm t}\in\mathbb R^3
    \mid
    \mathcal C(\mathbf q_{\rm t},\mathbf q_{\rm u})=\Gamma
    \right\},
\end{equation}
Unlike the Cassini oval, this \textit{Cram\'{e}r-Rao surface} is determined not only by the bistatic distance product, but also by the bistatic geometry and system observability. Similarly, when the target is fixed, $\{\mathbf q_{\rm u}\mid \mathcal C(\mathbf q_{\rm t},\mathbf q_{\rm u})=\Gamma\}$ characterizes the UAV-placement isosurface that achieves the same localization CRLB for the prescribed target.

\begin{figure}[t]
    \centering
    \subfloat[3D CRLB with fixed UAV.]{
        \includegraphics[width=0.45\columnwidth]{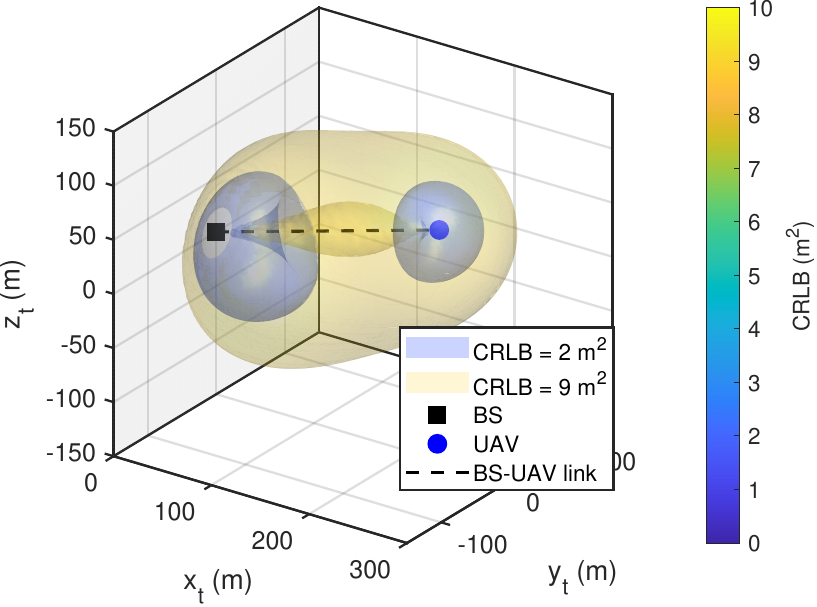}
        \label{fig:sim_3d_crlb_target}}
    \hfill
    \subfloat[3D CRLB with fixed target.]{
        \includegraphics[width=0.45\columnwidth]{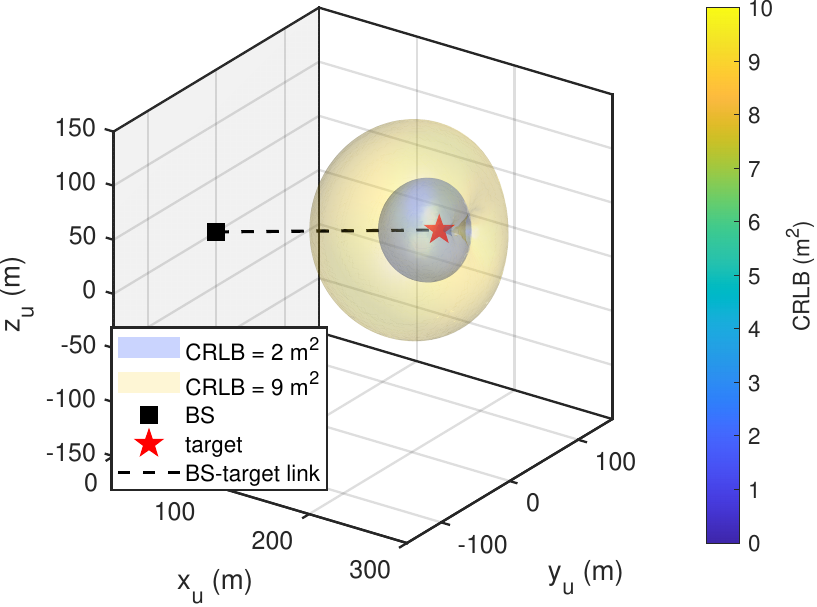}
        \label{fig:sim_3d_crlb_uav}}
    \caption{Plot of $\mathcal C(\mathbf q_{\rm t},\mathbf q_{\rm u})$, by varying the target location $\mathbf q_{\rm t}$, with $\mathbf q_{\rm u}$ fixed to $(100,0)$; and by varying the UAV location $\mathbf q_{\rm u}$, with $\mathbf q_{\rm t}$ fixed to $(100,40)$.}
    \label{fig:sim_3d_crlb_volume}
\end{figure}


\section{CRLB-Constrained Sensing Coverage}\label{sec:sensing_coverage}

The closed-form CRLB derived above characterizes the localization accuracy at a given target point. In practical UAV-assisted ISAC systems, however, the UAV is often deployed to provide reliable sensing over a spatial region rather than for a single known target point. Therefore, based on the definition of \textit{Cram\'{e}r-Rao surface}, we can define a more relevant metric termed the CRLB-constrained sensing coverage. For a fixed UAV location $\mathbf q_{\rm u}$ and a prescribed CRLB threshold $\Gamma>0$, the CRLB-constrained sensing coverage is defined as the set of target locations whose localization CRLB does not exceed $\Gamma$.
\begin{equation}
    \mathcal A_{\rm cov}(\mathbf q_{\rm u},\Gamma)
    \triangleq
    \left\{
    \mathbf q_{\rm t}
    \mid
    \mathcal C(\mathbf q_{\rm t},\mathbf q_{\rm u})\leq \Gamma
    \right\}.
    \label{eq:general_crlb_coverage}
\end{equation}

We first analyze the 2D coverage to expose the main boundary mechanisms in a transparent form, and then extend the obtained insights to the 3D case. In particular, for a relatively small $\Gamma$, the low-CRLB coverage is mainly concentrated around the BS and the UAV. This motivates a local boundary analysis around these two points. The symbol $\approx$ denotes the first-order local approximation as $r_{\rm t}/r_{\rm u}\to0$ in the BS-near region or $r_{\rm ut}/r_{\rm u}\to0$ in the UAV-near region.


\subsection{Two-Dimensional Sensing Coverage}\label{subsec:2d_coverage}

For a fixed UAV location $\mathbf q_{\rm u}=r_{\rm u}[\cos\phi_{\rm u},\sin\phi_{\rm u}]^{\rm T}$, the two-dimensional CRLB-constrained sensing coverage region is defined as
\begin{equation}\label{eq:2d_coverage_def}
    \mathcal A_{\rm cov}^{\rm 2D}(\mathbf q_{\rm u},\Gamma)
    \triangleq
    \left\{
    \mathbf q_{\rm t}\in\mathbb R^2
    \mid
    \mathcal C_{\rm 2D}(\mathbf q_{\rm t},\mathbf q_{\rm u})\le \Gamma
    \right\}.
\end{equation}
The exact boundary of $\mathcal A_{\rm cov}^{\rm 2D}$ is determined by the \textit{Cram\'{e}r-Rao curve} $\mathcal S_{\rm C}^{\rm 2D}(\Gamma;\mathbf q_{\rm u})\triangleq\left\{\mathbf q_{\rm t}\in\mathbb R^2\mid\mathcal C(\mathbf q_{\rm t},\mathbf q_{\rm u})=\Gamma\right\}$,
which generally does not admit a tractable global closed-form expression. Nevertheless, local closed-form boundaries can be obtained around the BS and the UAV. For a small CRLB threshold $\Gamma$, the coverage region is mainly concentrated around two local components: one near the BS and the other near the UAV, similar to that illustrated in Fig.~\ref{fig:sim_2d_crlb_heatmap}.

\subsubsection{BS-Near Coverage Boundary}

For the BS-near region, let $\mathbf q_{\rm t}=r_{\rm t}[\cos\phi_{\rm t},\sin\phi_{\rm t}]^{\rm T}$ with $r_{\rm t}\ll r_{\rm u}$. Then
\begin{equation}\label{eq:2d_bs_near_approx}
    r_{\rm ut}\approx r_{\rm u},
    \quad
    h\approx 1-\cos(\phi_{\rm t}-\phi_{\rm u}).
\end{equation}

\begin{proposition}[BS-near 2D coverage boundary]\label{prop:2d_bs_near_coverage}
In the BS-near region with $r_{\rm t}\ll r_{\rm u}$, the local CRLB satisfies
\begin{equation}\label{eq:2d_bs_near_crlb_prop}
    \mathcal C_{\rm 2D}^{\rm BS}(r_{\rm t},\phi_{\rm t})
    \approx
    A_{\rm BS}(\phi_{\rm t})r_{\rm t}^2
    +
    B_{\rm BS}(\phi_{\rm t})r_{\rm t}^4.
\end{equation}
where
\begin{equation}\label{eq:2d_A_BS}
    A_{\rm BS}(\phi_{\rm t})
    \triangleq
    \frac{c^2r_{\rm u}^2}
    {\xi_\tau\left(1-\cos(\phi_{\rm t}-\phi_{\rm u})\right)^2},
\end{equation}
\begin{equation}\label{eq:2d_B_BS}
    B_{\rm BS}(\phi_{\rm t})
    \triangleq
    \frac{
    2r_{\rm u}^2
    }
    {
    \xi_y
    \cos^2\phi_{\rm t}
    \left(1-\cos(\phi_{\rm t}-\phi_{\rm u})\right)
    }.
\end{equation}

Thus, the corresponding local coverage boundary is
\begin{equation}\label{eq:2d_bs_near_boundary_prop}
    r_{{\rm t},{\rm cov}}^{\rm BS}(\phi_{\rm t})
    \approx
    \left(
    \frac{
    -A_{\rm BS}(\phi_{\rm t})
    +
    \sqrt{
    A_{\rm BS}^2(\phi_{\rm t})
    +
    4B_{\rm BS}(\phi_{\rm t})\Gamma
    }
    }
    {2B_{\rm BS}(\phi_{\rm t})}
    \right)^{1/2}.
\end{equation}
\end{proposition}

\begin{IEEEproof}
Please refer to Appendix~\ref{app:bs_near_coverage}.
\end{IEEEproof}


\subsubsection{UAV-Near Coverage Boundary}

For the UAV-near region, let the target be locally represented as
\begin{equation}\label{eq:2d_uav_near_coordinate}
    \mathbf q_{\rm t}
    =
    \mathbf q_{\rm u}
    +
    r_{\rm ut}
    \begin{bmatrix}
        \cos\beta\\
        \sin\beta
    \end{bmatrix},
\end{equation}
where $r_{\rm ut}\ll r_{\rm u}$ and $\beta$ denotes the local direction from the UAV to the target. Then
\begin{equation}\label{eq:2d_uav_near_approx}
    r_{\rm t}\approx r_{\rm u},
    \quad
    h\approx1+\cos(\beta-\phi_{\rm u}).
\end{equation}

\begin{proposition}[UAV-near 2D coverage boundary]\label{prop:2d_uav_near_coverage}
In the UAV-near region with $r_{\rm ut}\ll r_{\rm u}$, the local CRLB satisfies
\begin{equation}\label{eq:2d_uav_near_crlb_prop}
\begin{aligned}
    \mathcal C_{\rm 2D}^{\rm UAV}(r_{\rm ut},\beta)
    &\approx
    r_{\rm ut}^2
    \left[
    \frac{c^2r_{\rm u}^2}
    {\xi_\tau\left(1+\cos(\beta-\phi_{\rm u})\right)^2}
    \right.
    \\
    &\quad\left.
    +
    \frac{2r_{\rm u}^4}
    {\xi_y\cos^2\phi_{\rm u}
    \left(1+\cos(\beta-\phi_{\rm u})\right)}
    \right].
\end{aligned}
\end{equation}
Thus, the corresponding local coverage boundary is
\begin{equation}\label{eq:2d_uav_near_boundary_prop}
    r_{{\rm ut},{\rm cov}}^{\rm UAV}(\beta)
    \approx
    \left[
    \frac{\Gamma}
    {
    \frac{c^2r_{\rm u}^2}
    {\xi_\tau\left(1+\cos(\beta-\phi_{\rm u})\right)^2}
    +
    \frac{2r_{\rm u}^4}
    {\xi_y\cos^2\phi_{\rm u}
    \left(1+\cos(\beta-\phi_{\rm u})\right)}
    }
    \right]^{1/2}.
\end{equation}
\end{proposition}

\begin{IEEEproof}
Please refer to Appendix~\ref{app:uav_near_coverage}.
\end{IEEEproof}


\subsection{3D Local Sensing Coverage}\label{subsec:3d_coverage}

We next extend the local coverage analysis to the general 3D case. For a fixed UAV location $\mathbf q_{\rm u}$ and a prescribed CRLB threshold $\Gamma>0$, the 3D CRLB-constrained sensing coverage region is defined as
\begin{equation}\label{eq:3d_coverage_def}
    \mathcal A_{\rm cov}^{\rm 3D}(\mathbf q_{\rm u},\Gamma)
    \triangleq
    \left\{
    \mathbf q_{\rm t}\in\mathbb R^3
    \mid
    \mathcal C(\mathbf q_{\rm t},\mathbf q_{\rm u})\le \Gamma
    \right\}.
\end{equation}
The exact global boundary of $\mathcal A_{\rm cov}^{\rm 3D}$ is determined by the \textit{Cram\'{e}r-Rao surface} in \eqref{eq:crlb_surface}, which generally does not admit a tractable closed form. Similar to the two-dimensional case, the 3D low-CRLB coverage can also be locally interpreted as the union of a BS-near component and a UAV-near component, as illustrated in Fig.~\ref{fig:sim_3d_crlb_volume}. The BS-near component forms a dimpled-cap-shaped body due to the bistatic collinearity direction, while the UAV-near component forms a cardioid-like body with a cusp toward the BS.
Define the UAV direction as $\mathbf e_{\rm u}=\mathbf q_{\rm u}/r_{\rm u}$, the BS-target direction as $\mathbf e_{\rm t}=\mathbf q_{\rm t}/r_{\rm t}$, and the UAV-target direction as $\mathbf e_{\rm ut}=(\mathbf q_{\rm t}-\mathbf q_{\rm u})/r_{\rm ut}$.

\subsubsection{BS-Near Coverage Boundary}

For the BS-near region, we have
\begin{equation}
    r_{\rm ut}\approx r_{\rm u},
    \quad
    \mathbf e_{\rm ut}
    =
    \frac{\mathbf q_{\rm t}-\mathbf q_{\rm u}}{r_{\rm ut}}
    \approx
    -\mathbf e_{\rm u}.
\end{equation}

Let $\chi_{\rm B}$ denote the angle between the BS-target direction and the BS-UAV direction, i.e.,
\begin{equation}\label{eq:chi_B_general}
\begin{aligned}
    \cos\chi_{\rm B}
    =
    \mathbf e_{\rm t}^{\rm T}\mathbf e_{\rm u}.
\end{aligned}
\end{equation}

\begin{proposition}[BS-near 3D coverage boundary]\label{prop:3d_bs_near_coverage}
In the BS-near region with $r_{\rm t}\ll r_{\rm u}$, the local CRLB satisfies
\begin{equation}\label{eq:bs_near_crlb_general_3d_prop}
    \mathcal C^{\rm BS}(\mathbf q_{\rm t},\mathbf q_{\rm u})
    \approx
    A_{\rm B}(\theta_{\rm t},\phi_{\rm t})r_{\rm t}^2
    +
    B_{\rm B}(\theta_{\rm t},\phi_{\rm t})r_{\rm t}^4,
\end{equation}
where
\begin{equation}\label{eq:A_B_general_3d_prop}
    A_{\rm B}(\theta_{\rm t},\phi_{\rm t})
    \triangleq
    \frac{c^2r_{\rm u}^2}
    {\xi_\tau(1-\cos\chi_{\rm B})^2},
\end{equation}

\begin{equation}\label{eq:B_B_general_3d_prop}
\begin{aligned}
    B_{\rm B}(\theta_{\rm t},\phi_{\rm t})
    &\triangleq
    \frac{r_{\rm u}^2}{\ell_x^2}
    \left[
    \frac{1-\ell_z^2}{\xi_y}
    +
    \frac{1-\ell_y^2}{\xi_z}
    \right.
    \\
    &\quad\left.
    +
    \frac{1}{(1-\cos\chi_{\rm B})^2}
    \left(
    \frac{\left(\Delta_{xy}^{\rm B}\right)^2}{\xi_y}
    +
    \frac{\left(\Delta_{xz}^{\rm B}\right)^2}{\xi_z}
    \right)
    \right].
\end{aligned}
\end{equation}
and
\begin{equation}\label{eq:bs_near_projected_factor_general}
\begin{aligned}
    \Delta_{xy}^{\rm B}
    &\triangleq
    \cos\theta_{\rm t}\cos\theta_{\rm u}
    \sin(\phi_{\rm u}-\phi_{\rm t}),
    \\
    \Delta_{xz}^{\rm B}
    &\triangleq
    \sin\theta_{\rm t}\cos\theta_{\rm u}\cos\phi_{\rm u}
    -
    \cos\theta_{\rm t}\sin\theta_{\rm u}\cos\phi_{\rm t}.
\end{aligned}
\end{equation}
The corresponding local coverage boundary is
\begin{equation}\label{eq:bs_near_boundary_general_3d_prop}
\begin{aligned}
    &r_{{\rm t},{\rm cov}}^{\rm BS}(\theta_{\rm t},\phi_{\rm t})
    \\
    &\approx
    \left(
    \frac{
    -A_{\rm B}(\theta_{\rm t},\phi_{\rm t})
    +
    \sqrt{
    A_{\rm B}^2(\theta_{\rm t},\phi_{\rm t})
    +
    4B_{\rm B}(\theta_{\rm t},\phi_{\rm t})\Gamma
    }
    }
    {2B_{\rm B}(\theta_{\rm t},\phi_{\rm t})}
    \right)^{1/2}.
\end{aligned}
\end{equation}
\end{proposition}

\begin{IEEEproof}
Please refer to Appendix~\ref{app:bs_near_coverage}.
\end{IEEEproof}


\subsubsection{UAV-Near Coverage Boundary}

For the UAV-near region, we have
\begin{equation}
    r_{\rm t}\approx r_{\rm u},
    \quad
    \mathbf e_{\rm t}
    \approx
    \mathbf e_{\rm u}.
\end{equation}
Let $\chi_{\rm U}$ denote the angle between the local UAV-target direction and the BS-UAV direction, i.e.,
\begin{equation}\label{eq:chi_U_general}
\begin{aligned}
    \cos\chi_{\rm U}
    =
    \mathbf e_{\rm ut}^{\rm T}\mathbf e_{\rm u}.
\end{aligned}
\end{equation}

\begin{proposition}[UAV-near 3D coverage boundary]\label{prop:3d_uav_near_coverage}
In the UAV-near region with $r_{\rm ut}\ll r_{\rm u}$, the local CRLB satisfies
\begin{equation}\label{eq:uav_near_crlb_general_3d_prop}
    \mathcal C^{\rm UAV}(\mathbf q_{\rm t},\mathbf q_{\rm u})
    \approx
    G_{\rm U}(\theta_{\rm ut},\phi_{\rm ut})r_{\rm ut}^2,
\end{equation}
where
\begin{equation}\label{eq:G_U_general_3d_prop}
\begin{aligned}
    &G_{\rm U}(\theta_{\rm ut},\phi_{\rm ut})
    \triangleq
    \frac{c^2r_{\rm u}^2}
    {\xi_\tau(1+\cos\chi_{\rm U})^2}
    +
    \frac{r_{\rm u}^4}{\ell_{x,{\rm u}}^2}
    \left[
    \frac{1-\ell_{z,{\rm u}}^2}{\xi_y}\right.
    \\
    &\quad+
    \frac{1-\ell_{y,{\rm u}}^2}{\xi_z}
    \left.
    +
    \frac{1}{(1+\cos\chi_{\rm U})^2}
    \left(
    \frac{\left(\Delta_{xy}^{\rm U}\right)^2}{\xi_y}
    +
    \frac{\left(\Delta_{xz}^{\rm U}\right)^2}{\xi_z}
    \right)
    \right].
\end{aligned}
\end{equation}
\begin{equation}\label{eq:uav_near_projected_factor_general}
\begin{aligned}
    \Delta_{xy}^{\rm U}
    &\triangleq
    \cos\theta_{\rm ut}\cos\theta_{\rm u}
    \sin(\phi_{\rm u}-\phi_{\rm ut}),
    \\
    \Delta_{xz}^{\rm U}
    &\triangleq
    \sin\theta_{\rm ut}\cos\theta_{\rm u}\cos\phi_{\rm u}
    -
    \cos\theta_{\rm ut}\sin\theta_{\rm u}\cos\phi_{\rm ut}.
\end{aligned}
\end{equation}
and $\ell_{x,{\rm u}}=\cos\theta_{\rm u}\cos\phi_{\rm u}$, $\ell_{y,{\rm u}}=\cos\theta_{\rm u}\sin\phi_{\rm u}$, $\ell_{z,{\rm u}}=\sin\theta_{\rm u}$.
The corresponding local coverage boundary is
\begin{equation}\label{eq:uav_near_boundary_general_3d_prop}
    r_{{\rm ut},{\rm cov}}^{\rm U}(\theta_{\rm ut},\phi_{\rm ut})
    \approx
    \left(
    \frac{\Gamma}
    {G_{\rm U}(\theta_{\rm ut},\phi_{\rm ut})}
    \right)^{1/2}.
\end{equation}
\end{proposition}

\begin{IEEEproof}
Please refer to Appendix~\ref{app:uav_near_coverage}.
\end{IEEEproof}



\subsection{CRLB-Constrained Coverage Size}\label{subsec:coverage_size}

Beyond the local boundary shape, the CRLB-constrained coverage can also be quantified by its area or volume. This provides a direct placement metric for UAV-assisted ISAC. 
Based on \eqref{eq:general_crlb_coverage}, in the two-dimensional case, the corresponding coverage area is
\begin{equation}
    A_{\rm cov}(\mathbf q_{\rm u},\Gamma)
    =
    \int_{\mathcal D}
    \mathbb I\left\{
    \mathcal C_{\rm 2D}(\mathbf q_{\rm t},\mathbf q_{\rm u})\le \Gamma
    \right\}
    d\mathbf q_{\rm t},
\end{equation}
where $\mathcal D$ denotes the target evaluation region and $\mathbb I\{\cdot\}$ is the indicator function. Similarly, in the 3D case, the coverage volume is given by
\begin{equation}
    V_{\rm cov}(\mathbf q_{\rm u},\Gamma)
    =
    \int_{\mathcal D}
    \mathbb I\left\{
    \mathcal C_{\rm 3D}(\mathbf q_{\rm t},\mathbf q_{\rm u})\le \Gamma
    \right\}
    d\mathbf q_{\rm t}.
\end{equation}

For UAV-proximal sensing tasks, the total coverage size may not fully reflect the useful sensing region around the UAV. Therefore, we further define the UAV-side coverage as the subset of covered target points that are closer to the UAV than to the BS. In the two-dimensional case, it is given by
\begin{equation}
    \mathcal A_{\rm UAV}^{\rm side}(\mathbf q_{\rm u},\Gamma)
    =
    \left\{
    \mathbf q_{\rm t}\in\mathcal A_{\rm cov}(\mathbf q_{\rm u},\Gamma):
    \|\mathbf q_{\rm t}-\mathbf q_{\rm u}\|<\|\mathbf q_{\rm t}\|
    \right\},
\end{equation}
and its area is
\begin{equation}
    A_{\rm UAV}^{\rm side}(\mathbf q_{\rm u},\Gamma)
    =
    \int_{\mathcal D}
    \mathbb I\left\{
    \mathbf q_{\rm t}\in\mathcal A_{\rm UAV}^{\rm side}(\mathbf q_{\rm u},\Gamma)
    \right\}
    d\mathbf q_{\rm t}.
\end{equation}
The 3D UAV-side coverage volume $V_{\rm UAV}^{\rm side}(\mathbf q_{\rm u},\Gamma)$ can be defined in the same manner by replacing the area measure with the volume measure. This geometric partition is used instead of connected-component decomposition, because it remains well defined even when the BS-near and UAV-near coverage components merge into a single connected region.

The above definitions also lead to a coverage-oriented UAV placement criterion. For example, the UAV placement that maximizes the total CRLB-constrained coverage can be formulated as
\begin{equation}
    \mathbf q_{\rm u}^{\rm cov}
    \in
    \arg\max_{\mathbf q_{\rm u}\in\mathcal Q}
    A_{\rm cov}(\mathbf q_{\rm u},\Gamma)
\end{equation}
in the two-dimensional case, or with $V_{\rm cov}(\mathbf q_{\rm u},\Gamma)$ in the 3D case. If the mission focuses on local sensing around the UAV, $A_{\rm UAV}^{\rm side}(\mathbf q_{\rm u},\Gamma)$ or $V_{\rm UAV}^{\rm side}(\mathbf q_{\rm u},\Gamma)$ can be used instead. The total coverage size and the UAV-side coverage size correspond to two different placement objectives. The former is suitable for maximizing the overall reliable sensing region, while the latter is more relevant when the UAV is expected to approach and sense a local region of interest. Hence, the same CRLB expression can support different UAV placement criteria depending on the sensing mission.

\begin{figure*}[t]
    \centering
    \subfloat[Selected UAV positions.]{
        \includegraphics[width=0.31\textwidth]{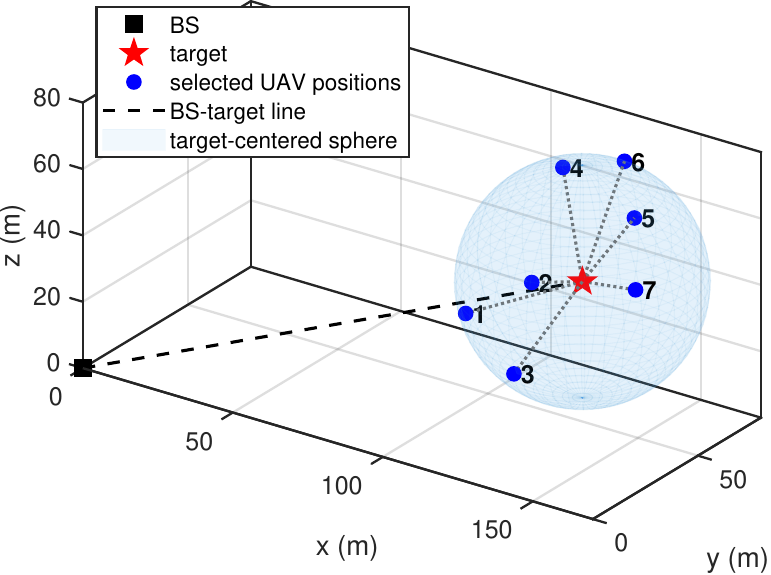}
        \label{fig:sim_exp6_sphere_geometry}}
    \hfill
    \subfloat[ML MSE and OFDM-CRLB.]{
        \includegraphics[width=0.31\textwidth]{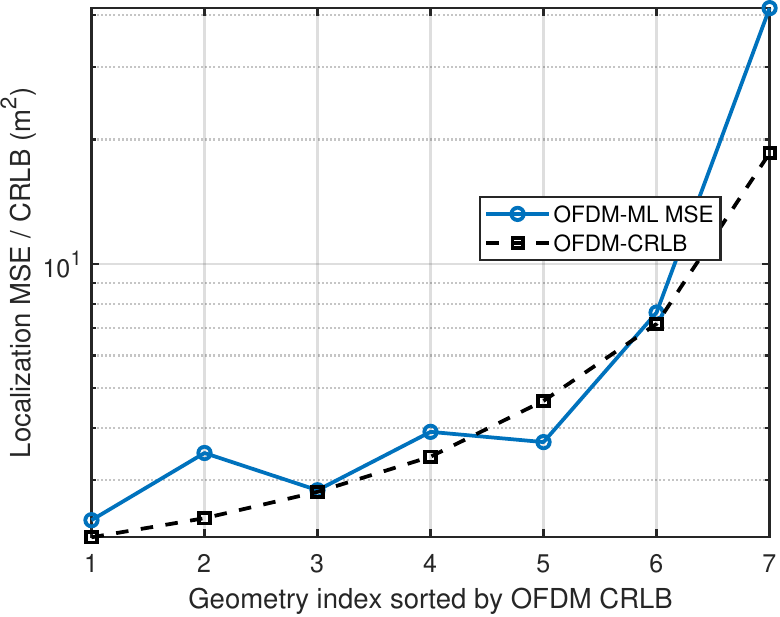}
        \label{fig:sim_exp6_sphere_mse_crlb}}
    \hfill
    \subfloat[CRLB components.]{
        \includegraphics[width=0.31\textwidth]{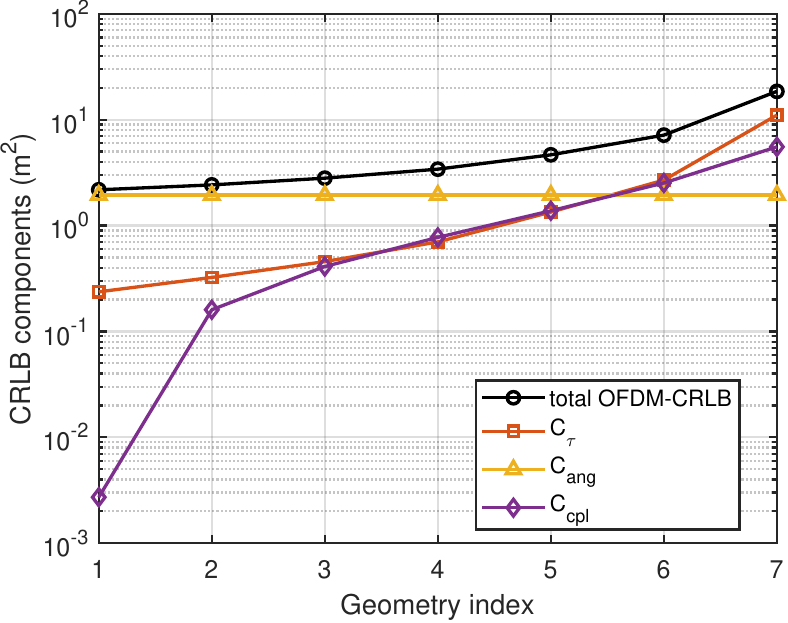}
        \label{fig:sim_exp6_sphere_decomp}}
    \caption{OFDM-ML validation under fixed target and fixed UAV-target distance. The UAV positions are selected on a target-centered sphere and sorted by OFDM-CRLB, so that the received sensing SNR is fixed while the 3D geometry changes.}
    \label{fig:sim_exp6_sphere}
\end{figure*}

\section{OFDM Specialization and Numerical Validation}\label{sec:verification}

In this section, we first specialize the derived closed-form CRLB to a representative OFDM-based ISAC waveform, and then provide numerical results to validate the analysis and illustrate the general sensing-coverage behavior. Unless otherwise stated, the simulation parameters used throughout this section are listed in Table~\ref{tab:sim_params}. To validate the analytical CRLB in Theorem \ref{the:3d_crlb}, we specialize the results by considering OFDM as the basic sensing waveform, and construct the corresponding ML sensing algorithm, which asymptotically achieves the CRLB. Specifically, consider an OFDM sensing waveform with $K$ subcarriers, $M$ OFDM symbols, subcarrier spacing $\Delta f$, and uniform power allocation scheme $P=P_{\rm t}/K$. The transmitted resource-element symbols are assumed to be known at the BS for sensing processing. Let $\sigma^2=N_0\Delta f$ denote the noise power on each subcarrier. The sensing SNR on each resource element is
\begin{equation}
    \gamma_{{\rm s},{\rm RE}}
    =
    \frac{P|\alpha_{\rm t}|^2N_{\rm R}}{\sigma^2}.
\end{equation}
Since the useful sensing observation contains $MK$ resource elements, the accumulated sensing observation SNR is
\begin{equation}
    \Upsilon_{\rm s}^{\rm OFDM}
    =
    MK\gamma_{{\rm s},{\rm RE}}.
\end{equation}
For equal-energy subcarriers, the RMS bandwidth of the OFDM waveform is
\begin{equation}
    \beta_{\rm rms,OFDM}^2
    =
    \frac{(K^2-1)\Delta f^2}{12}.
\end{equation}
Accordingly, the delay Fisher information becomes
\begin{equation}
\begin{aligned}
    I_{\tau_{\rm t}}^{\rm OFDM}
    &=
    8\pi^2\beta_{\rm rms,OFDM}^2
    \Upsilon_{\rm s}^{\rm OFDM}\\
    &=
    \frac{2\pi^2}{3}
    MK(K^2-1)\Delta f^2\gamma_{{\rm s},{\rm RE}}.
\end{aligned}
\end{equation}
The angular Fisher information is obtained by replacing $\Upsilon_{\rm s}$ in Lemma~\ref{lem:angular_fim} with $\Upsilon_{\rm s}^{\rm OFDM}$. Therefore, the OFDM-based CRLB is obtained by substituting the corresponding OFDM information coefficients into $\xi_\tau$, $\xi_y$, and $\xi_z$, while the geometry-dependent factors $h$, $\ell_x$, $S_{xy}$, and $S_{xz}$ remain the same. In the following simulations, this OFDM-specialized CRLB is compared with the localization MSE obtained by an ML estimator with local Cartesian search.

\subsection{CRLB Verification}
To isolate the geometry-dependent effect from the received sensing SNR, Fig.~\ref{fig:sim_exp6_sphere} fixes the target location and the UAV-target distance, and selects UAV positions on a sphere centered at the target. The number of subcarriers, symbols and subcarrier spacing are set as $K=64$, $M=16$ and $\Delta f=120$ kHz, respectively. The selected UAV positions are sorted by the resulting OFDM-CRLB. Since both $r_{\rm t}$ and $r_{\rm ut}$ are fixed, the accumulated sensing SNR remains unchanged across these geometries. Nevertheless, the CRLB varies substantially, which indicates that the 3D bistatic geometry alone can significantly reshape the localization bound. In particular, the last selected geometry is close to the bistatic collinearity singularity, where the factor $h$ becomes small and the CRLB increases sharply. The ML MSE follows the CRLB trend for the regular geometries, while the near-singular geometry leads to a very large theoretical bound. Fig.~\ref{fig:sim_exp6_sphere_decomp} further explains this behavior through the CRLB decomposition. In the sphere sweep, the angular term $C_{\rm ang}$ remains almost unchanged because the target direction and the UAV-target distance are fixed. By contrast, the delay-related term $C_\tau$ increases rapidly as the geometry approaches the small-$h$ condition, and the coupling term $C_{\rm cpl}$ also becomes more pronounced near the ill-conditioned geometry. 

Fig.~\ref{fig:sim_exp6_traj} further considers a more practical straight-line UAV trajectory. Along this trajectory, the UAV moves closer to the target while maintaining a nonzero angle with the BS-target link. Therefore, both the accumulated sensing SNR and the bistatic geometry vary with the trajectory index. The OFDM-CRLB decreases along the trajectory, and the ML MSE follows the same overall trend with a moderate gap. This shows that the proposed CRLB remains predictive under a practical UAV movement pattern where propagation loss and geometry are coupled. The decomposition in Fig.~\ref{fig:sim_exp6_traj_decomp} shows a different mechanism from the sphere sweep. Along the straight-line trajectory, the bistatic factor $h$ stays away from the singular condition, and the coupling term $C_{\rm cpl}$ remains relatively small. The CRLB reduction is therefore mainly attributed to the decreasing UAV-target distance and the corresponding reductions of the delay and angular contributions.

\begin{figure*}[t]
    \centering
    \subfloat[Straight-line UAV trajectory.]{
        \includegraphics[width=0.31\textwidth]{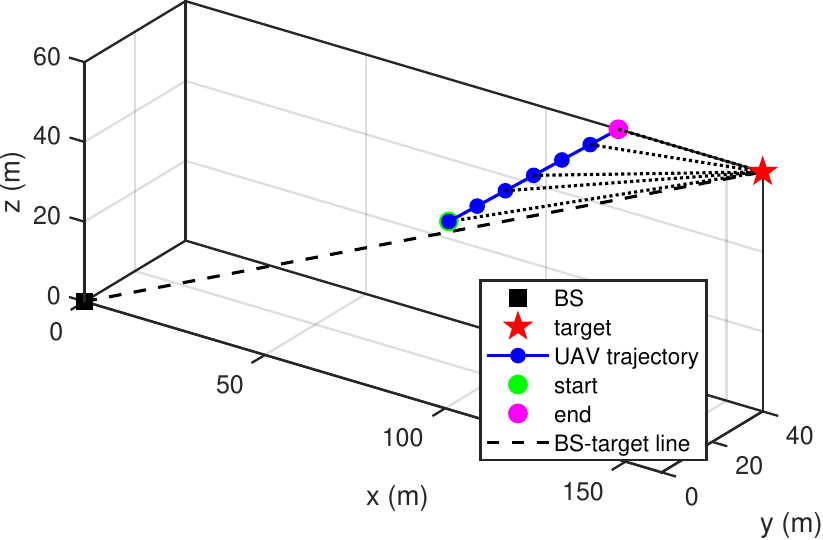}
        \label{fig:sim_exp6_traj_geometry}}
    \hfill
    \subfloat[ML MSE and OFDM-CRLB.]{
        \includegraphics[width=0.31\textwidth]{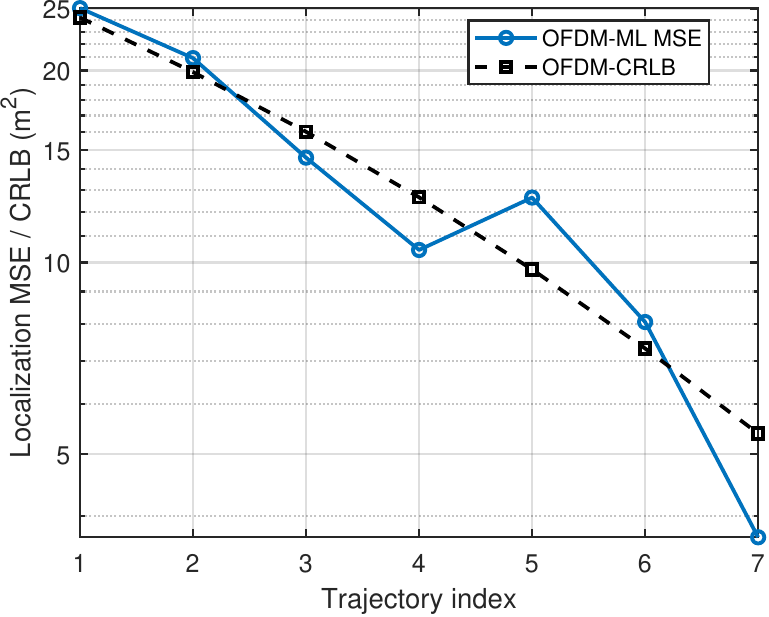}
        \label{fig:sim_exp6_traj_mse_crlb}}
    \hfill
    \subfloat[CRLB components.]{
        \includegraphics[width=0.31\textwidth]{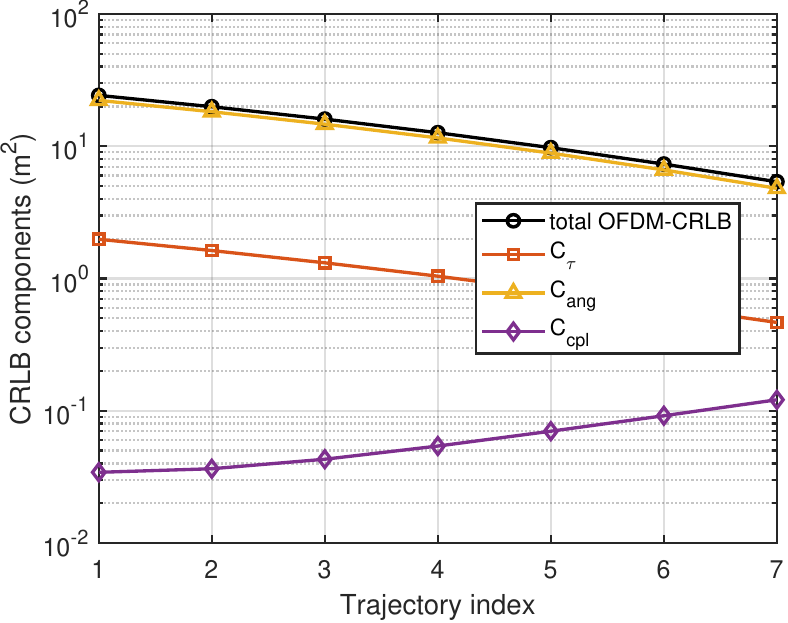}
        \label{fig:sim_exp6_traj_decomp}}
    \caption{OFDM-ML validation along a straight-line UAV trajectory. Along this more practical UAV movement pattern, both the received sensing SNR and the bistatic geometry vary, and the ML MSE follows the placement-dependent OFDM-CRLB trend.}
    \label{fig:sim_exp6_traj}
\end{figure*}

\subsection{CRLB-constrained coverage}
To complement the above CRLB-constrained coverage analysis, we provide a numerical illustration of the exact coverage boundary and its local approximation. The exact coverage is obtained by directly evaluating the closed-form CRLB over the target region, while the local boundaries are generated from Propositions~\ref{prop:2d_bs_near_coverage}-\ref{prop:3d_uav_near_coverage}. 
\begin{figure}[t]
    \centering
    \subfloat[2D boundary, $\Gamma=1$ m$^2$.]{
        \includegraphics[width=0.45\columnwidth]{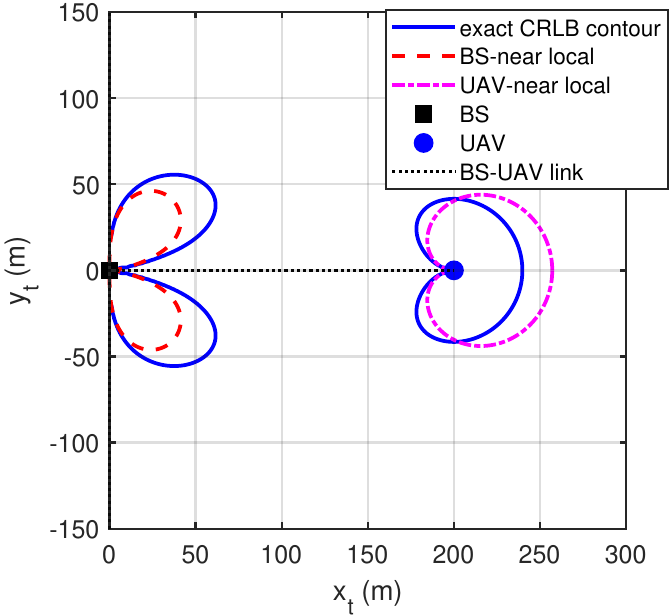}
        \label{fig:coverage_approx_2d}}
    \hfill
    \subfloat[3D boundary, $\Gamma=1$ m$^2$.]{
        \includegraphics[width=0.45\columnwidth]{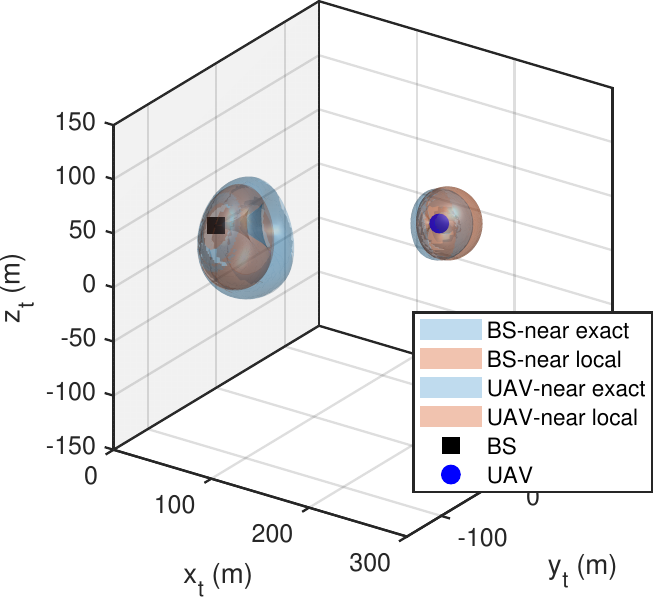}
        \label{fig:coverage_approx_3d}}
    \caption{Representative comparisons between the exact CRLB-constrained coverage boundary and the corresponding local approximation in the two-dimensional and 3D cases.}
    \label{fig:coverage_approx_comparison}
\end{figure}
As shown in Fig.~\ref{fig:coverage_approx_comparison}, the local formulas capture the dominant geometric shapes of the exact CRLB level sets around their corresponding expansion points. In the two-dimensional case, the BS-near boundary is squeezed by the BS-UAV collinearity direction and the array endfire direction, while the UAV-near boundary forms a heart-like region with a cusp toward the BS, which is consistent with the factor $1+\cos(\beta-\phi_{\rm u})$ in Proposition~\ref{prop:2d_uav_near_coverage}. In the 3D case, the BS-near surface exhibits a dimpled-cap-shaped body, where the indentation is aligned with the bistatic collinearity direction. This agrees with the role of $1-\cos\chi_{\rm B}$ in Proposition~\ref{prop:3d_bs_near_coverage}. In addition, the UAV-near surface exhibits a cardioid-like body with a cusp toward the BS, which is governed by the factor $1+\cos\chi_{\rm U}$ in Proposition~\ref{prop:3d_uav_near_coverage}. The results show that the local approximations provide interpretable geometric descriptions of the exact CRLB-constrained coverage.

We further illustrate the CRLB-constrained coverage size induced by different UAV placements. Fig.~\ref{fig:sim_exp7_2d} first shows the two-dimensional coverage-size heatmaps when $\Gamma=1$ m$^2$. The total coverage area has a broader high-value region, whereas the UAV-side coverage area is more concentrated near the UAV-proximal sensing region. Their maximizing UAV locations are different, indicating that global coverage and UAV-side coverage correspond to different placement objectives.

\begin{figure}[t]
    \centering
    \subfloat[Total coverage area.]{
        \includegraphics[width=0.45\columnwidth]{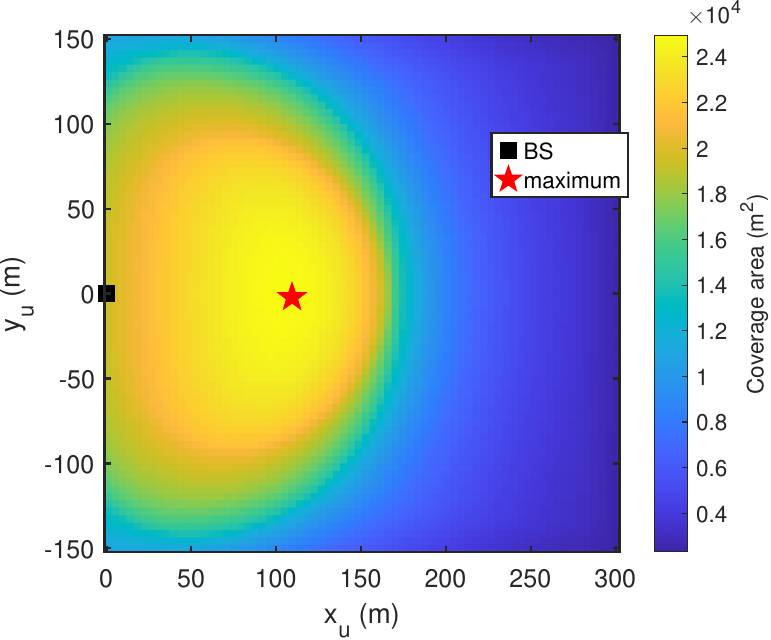}
        \label{fig:sim_exp7_2d_total}}
    \hfill
    \subfloat[UAV-side coverage area.]{
        \includegraphics[width=0.45\columnwidth]{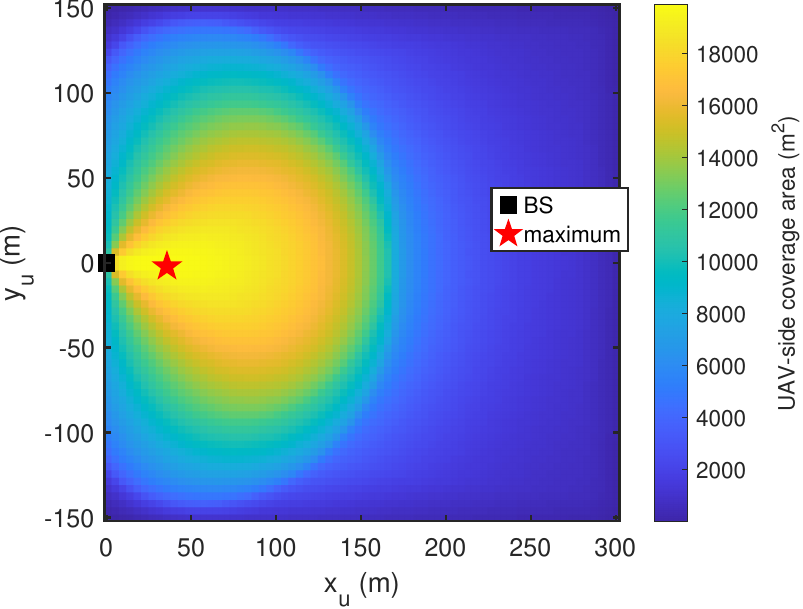}
        \label{fig:sim_exp7_2d_uavside}}
    \caption{Two-dimensional CRLB-constrained coverage area under different UAV placements.}
    \label{fig:sim_exp7_2d}
\end{figure}

To more clearly show the variation with respect to UAV displacement, Fig.~\ref{fig:sim_exp7_2d_curves} plots the coverage area along $y_{\rm u}=0$ under different CRLB thresholds. The curves exhibit a non-monotonic trend: the coverage first increases as the UAV moves away from the BS, reaches a finite-distance maximum, and then decreases when the UAV moves farther away. A looser threshold enlarges the coverage area and shifts the useful range outward. Moreover, the total-area optimum appears at a larger $x_{\rm u}$ than the UAV-side optimum.

\begin{figure}[t]
    \centering
    \subfloat[Total coverage area.]{
        \includegraphics[width=0.45\columnwidth]{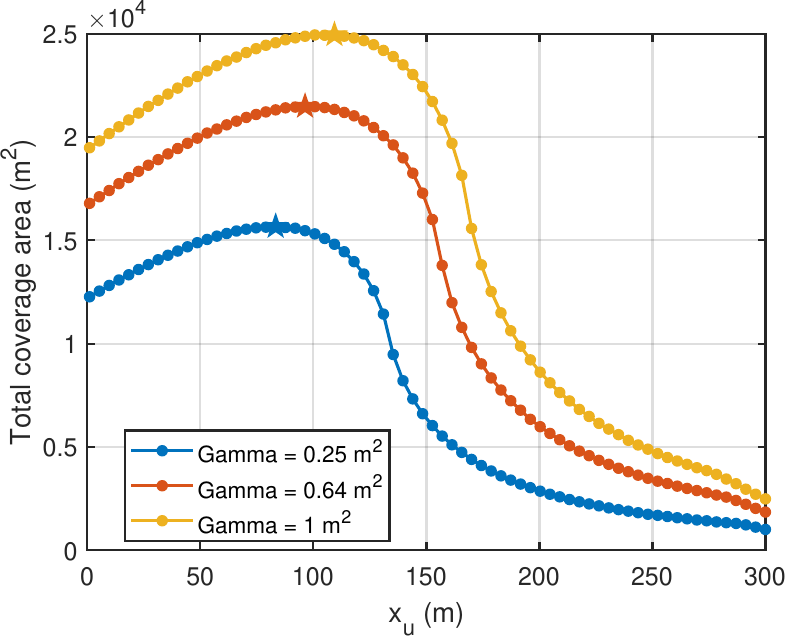}
        \label{fig:sim_exp7_2d_total_curve}}
    \subfloat[UAV-side coverage area.]{
        \includegraphics[width=0.45\columnwidth]{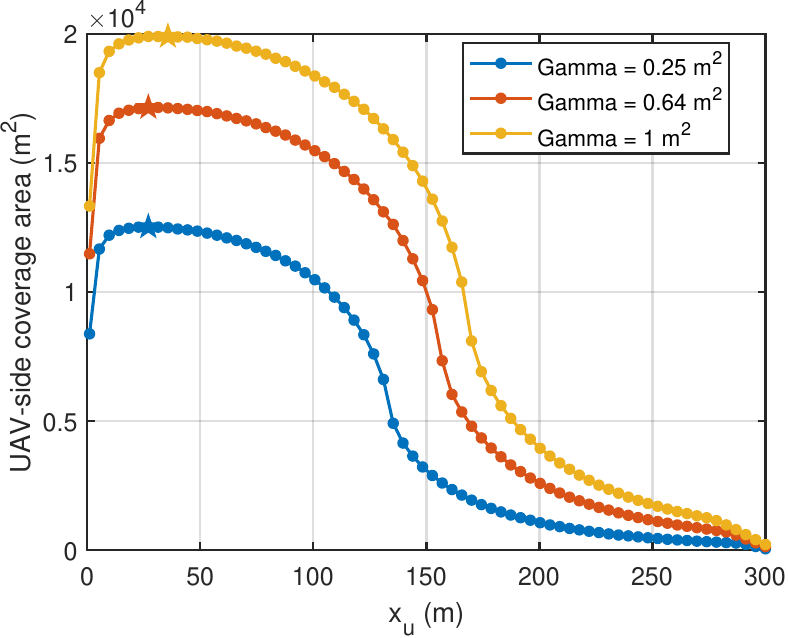}
        \label{fig:sim_exp7_2d_uavside_curve}}
    \caption{Two-dimensional CRLB-constrained coverage area versus $x_{\rm u}$ along $y_{\rm u}=0$ under different CRLB thresholds.}
    \label{fig:sim_exp7_2d_curves}
\end{figure}

Fig.~\ref{fig:sim_exp7_3d_curves} further shows the 3D coverage volume along $y_{\rm u}=0$ with a fixed $\Gamma=1$ m$^2$. The total coverage volume is also non-monotonic in $x_{\rm u}$, but its peak value and maximizing location vary with altitude. Increasing the UAV altitude generally reduces the achievable volume and shifts the preferred horizontal placement. The UAV-side volume is even more sensitive to altitude: the high-volume region shrinks at larger altitude, showing that excessive altitude weakens UAV-proximal sensing coverage significantly. It should be noted that this altitude-dependent trend is scenario-specific rather than a general property. In the considered setup, the BS is located at the origin and the target evaluation region is the whole space. Hence, increasing $z_{\rm u}$ enlarges the UAV-target distance for most candidate target points and leads to a smaller coverage volume. For a sensing region of interest located at a higher altitude, or for a different UAV feasible region, an intermediate or higher UAV altitude may become preferable. Therefore, Fig.~\ref{fig:sim_exp7_3d_curves} should be interpreted as a coverage-oriented placement insight under the considered setup, rather than as a universal claim that the coverage volume always decreases with UAV altitude.

\begin{figure}[t]
    \centering
    \subfloat[Total coverage volume.]{
        \includegraphics[width=0.45\columnwidth]{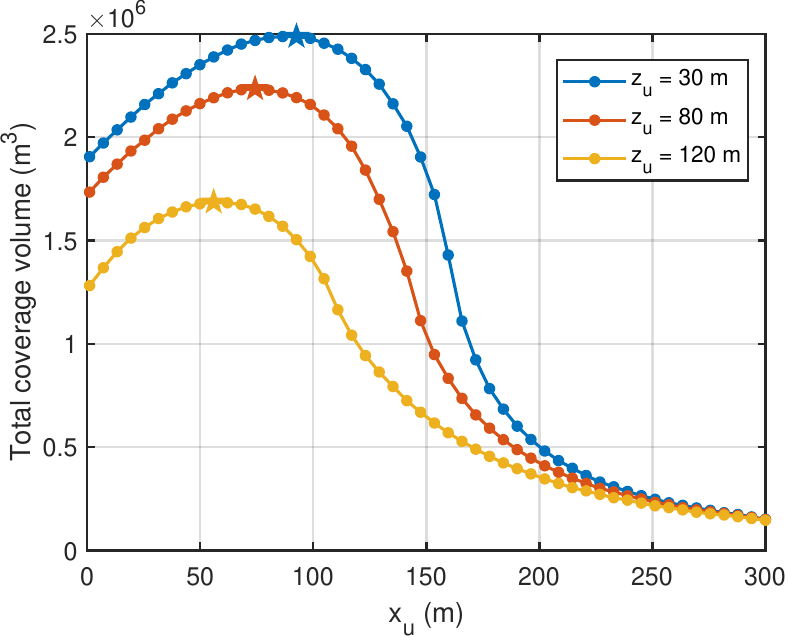}
        \label{fig:sim_exp7_3d_total_curve}}
    \subfloat[UAV-side coverage volume.]{
        \includegraphics[width=0.45\columnwidth]{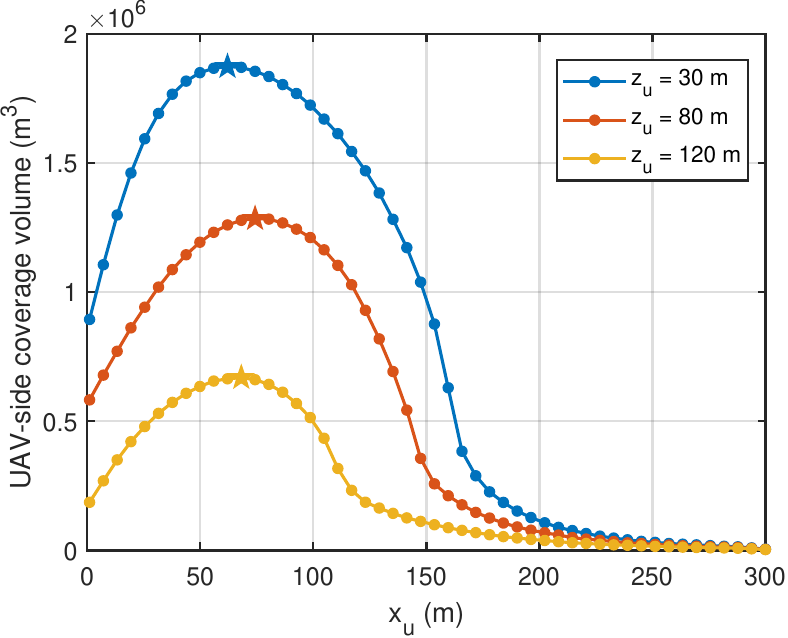}
        \label{fig:sim_exp7_3d_uavside_curve}}
    \caption{3D CRLB-constrained coverage volume versus $x_{\rm u}$ along $y_{\rm u}=0$ under different UAV altitudes.}
    \label{fig:sim_exp7_3d_curves}
\end{figure}

\section{Conclusion}\label{sec:conclusion}
This paper investigated the fundamental performance limits of 3D target localization in UAV-assisted bistatic ISAC systems. A closed-form 3D localization CRLB was derived as an explicit function of both the UAV anchor location and the target position. The derived expression decomposes the localization error bound into three components associated with propagation delay, angular measurements, and their coupling effect. The CRLB was validated through an OFDM-based ISAC waveform with ML estimation under different UAV deployments, showing that the analytical bound accurately captures placement-dependent localization trends. Based on the closed-form CRLB, we introduced and analyzed CRLB-constrained sensing coverage, which characterizes the spatial region where a prescribed localization-accuracy requirement can be guaranteed. Through local boundary approximations and coverage area/volume evaluations, we showed how UAV displacement, altitude, and the CRLB threshold reshape the reliable sensing region. These results provide insights for 3D localization-oriented UAV-assisted ISAC design.



\appendices

\section{Proof of Lemma~\ref{lem:angular_fim}}
\label{app:angle_information}
For the $(n_y,n_z)$-th antenna element, $0\le n_y\le N_y-1$ and $0\le n_z\le N_z-1$, the steering-vector entry is
\begin{equation}
a_{n_y,n_z}=e^{-j\pi(n_y\cos\theta_{\rm t}\sin\phi_{\rm t}
+n_z\sin\theta_{\rm t})}.
\end{equation}
For $\eta\in\{\theta_{\rm t},\phi_{\rm t}\}$, its derivative can be written as
\begin{equation}
\frac{\partial a_{n_y,n_z}}{\partial\eta}
=-j\pi w_{\eta,n_y,n_z}a_{n_y,n_z},
\end{equation}
where
\begin{equation}
w_{\theta_{\rm t},n_y,n_z}=
-n_y\sin\theta_{\rm t}\sin\phi_{\rm t}
+n_z\cos\theta_{\rm t},
\end{equation}
\begin{equation}
w_{\phi_{\rm t},n_y,n_z}=
n_y\cos\theta_{\rm t}\cos\phi_{\rm t}.
\end{equation}

Since $\alpha_{\rm t}$ is unknown, the component of $\partial\mathbf a/\partial\eta$ parallel to $\mathbf a$ cannot be distinguished from a change of the unknown complex gain. Thus, the effective angular Fisher information is 
\begin{equation} 
[\mathbf I_{\rm ang}]_{ij} = \frac{2P_{\rm t}T_{\rm obs}|\alpha_{\rm t}|^2}{N_0} \Re\left\{ \frac{\partial\mathbf a^{\rm H}}{\partial\eta_i} \mathbf P_{\mathbf a}^{\perp} \frac{\partial\mathbf a}{\partial\eta_j} \right\}, 
\label{eq:angular_fim_projection} 
\end{equation} 
where $\eta_i,\eta_j\in\{\theta_{\rm t},\phi_{\rm t}\}$ and 
\begin{equation} 
\mathbf P_{\mathbf a}^{\perp} = \mathbf I - \frac{\mathbf a\mathbf a^{\rm H}}{\|\mathbf a\|^2}. \end{equation}

The projection is equivalent to replacing the antenna indices by their centered versions. Therefore,
\begin{equation}
\begin{aligned}
&\frac{\partial\mathbf a^{\rm H}}{\partial\eta_i}
\mathbf P_{\mathbf a}^{\perp}
\frac{\partial\mathbf a}{\partial\eta_j}\\
&=\pi^2N_{\rm R}\left[
\frac{N_y^2-1}{12}u_{\eta_i}u_{\eta_j}
+\frac{N_z^2-1}{12}v_{\eta_i}v_{\eta_j}
\right],
\end{aligned}
\end{equation}
where $u_{\theta_{\rm t}}=-\sin\theta_{\rm t}\sin\phi_{\rm t}$, $v_{\theta_{\rm t}}=\cos\theta_{\rm t}$, $u_{\phi_{\rm t}}=\cos\theta_{\rm t}\cos\phi_{\rm t}$, and $v_{\phi_{\rm t}}=0$. Since 
\begin{equation} 
\Upsilon_{\rm s} = \frac{P_{\rm t}T_{\rm obs}|\alpha_{\rm t}|^2N_{\rm R}}{N_0} = \frac{\xi_1}{r_{\rm ut}^2r_{\rm t}^2}, 
\end{equation} 
we obtain 
\begin{equation} 
[\mathbf I_{\rm ang}]_{ij} = \frac{\pi^2\Upsilon_{\rm s}}{6} \left( (N_y^2-1)u_{\eta_i}u_{\eta_j} + (N_z^2-1)v_{\eta_i}v_{\eta_j} \right). 
\end{equation} 
Substituting the four coefficients $u_\eta$ and $v_\eta$ gives the desired result.

\section{Proof of Theorem~\ref{the:3d_crlb}}
\label{app:3d_crlb}
Define
\begin{equation}
    h\triangleq 1+\mathbf e_r^{\rm T}\mathbf e_{\rm ut}
    =
    \frac{(r_{\rm t}+r_{\rm ut})^2-r_{\rm u}^2}
    {2r_{\rm t}r_{\rm ut}},
    \quad
    a\triangleq\mathbf e_{\rm ut}^{\rm T}\mathbf e_\theta,
    \quad
    b\triangleq\mathbf e_{\rm ut}^{\rm T}\mathbf e_\phi .
\end{equation}
Consider the local orthonormal basis
$\{\mathbf e_r,\mathbf e_\theta,\mathbf e_\phi\}$ at the target location. The target displacement and the corresponding bistatic-delay differential are respectively given by
\begin{equation}
    {\rm d}\mathbf q_{\rm t}
    =
    \mathbf e_r{\rm d}r_{\rm t}
    +
    r_{\rm t}\mathbf e_\theta{\rm d}\theta_{\rm t}
    +
    r_{\rm t}\cos\theta_{\rm t}\mathbf e_\phi{\rm d}\phi_{\rm t},
\end{equation}
and
\begin{equation}
    {\rm d}\tau_{\rm t}
    =
    \frac{1}{c}
    \left(
    h\,{\rm d}r_{\rm t}
    +
    r_{\rm t}a\,{\rm d}\theta_{\rm t}
    +
    r_{\rm t}\cos\theta_{\rm t}b\,{\rm d}\phi_{\rm t}
    \right).
\end{equation}
Let
\begin{equation}
    A\triangleq I_{\theta_{\rm t}\theta_{\rm t}},
    \quad
    B\triangleq I_{\theta_{\rm t}\phi_{\rm t}},
    \quad
    D\triangleq I_{\phi_{\rm t}\phi_{\rm t}},
    \quad
    \Delta\triangleq AD-B^2.
\end{equation}
By block matrix inversion of the position-domain Fisher information matrix, the localization CRLB is
\begin{equation}
\begin{aligned}
\mathcal C(\mathbf q_{\rm t},\mathbf q_{\rm u})
={}&
\frac{c^2}{I_{\tau_{\rm t}}h^2}
+
\frac{r_{\rm t}^2}{\Delta}
\left(D+\cos^2\theta_{\rm t}A\right)
\\
&+
\frac{r_{\rm t}^2}{h^2\Delta}
\left(
Da^2
-
2Bab\cos\theta_{\rm t}
+
Ab^2\cos^2\theta_{\rm t}
\right).
\end{aligned}
\label{eq:crlb_proof_compact}
\end{equation}
From Lemma~\ref{lem:angular_fim}, we have
\begin{equation}
    \Delta
    =
    \frac{
    \xi_y\xi_z\cos^4\theta_{\rm t}\cos^2\phi_{\rm t}
    }{
    r_{\rm ut}^4r_{\rm t}^4
    }.
\end{equation}
Moreover, the spherical-coordinate representation gives
\begin{equation}
\begin{aligned}
    a
    &=
    \frac{r_{\rm u}}{r_{\rm ut}}
    \left(
    \cos\theta_{\rm u}\sin\theta_{\rm t}
    \cos(\phi_{\rm u}-\phi_{\rm t})
    -
    \sin\theta_{\rm u}\cos\theta_{\rm t}
    \right),\\
    b
    &=
    -\frac{r_{\rm u}}{r_{\rm ut}}
    \cos\theta_{\rm u}\sin(\phi_{\rm u}-\phi_{\rm t}),
\end{aligned}
\end{equation}
which further yields
\begin{equation}
\begin{aligned}
    a\cos\phi_{\rm t}
    +b\sin\theta_{\rm t}\sin\phi_{\rm t}
    =
    \frac{r_{\rm u}}{r_{\rm ut}}
    \Big(
    &\cos\theta_{\rm u}\sin\theta_{\rm t}\cos\phi_{\rm u}
    \\
    &-
    \sin\theta_{\rm u}\cos\theta_{\rm t}\cos\phi_{\rm t}
    \Big).
\end{aligned}
\end{equation}
Finally, substituting $I_{\tau_{\rm t}}$, $A$, $B$, $D$, $\Delta$, $h$, $a$, and $b$ into \eqref{eq:crlb_proof_compact}, followed by straightforward simplification, gives \eqref{eq:3d_crlb}.

\section{Proof of BS-Near Coverage Approximations}
\label{app:bs_near_coverage}

We first consider the BS-near region. Let $\mathbf q_{\rm t}=r_{\rm t}\mathbf e_{\rm t}$ and $\mathbf q_{\rm u}=r_{\rm u}\mathbf e_{\rm u}$, where $\mathbf e_{\rm t}$ and $\mathbf e_{\rm u}$ denote the BS-target and BS-UAV unit directions, respectively. When $r_{\rm t}\ll r_{\rm u}$, we have $r_{\rm ut}\approx r_{\rm u}$ and $\mathbf e_{\rm ut}=(\mathbf q_{\rm t}-\mathbf q_{\rm u})/r_{\rm ut}\approx-\mathbf e_{\rm u}$. Therefore,
\begin{equation}
    h=1+\mathbf e_{\rm t}^{\rm T}\mathbf e_{\rm ut}
    \approx
    1-\mathbf e_{\rm t}^{\rm T}\mathbf e_{\rm u}.
\end{equation}

For the two-dimensional case, $\mathbf e_{\rm t}=[\cos\phi_{\rm t},\sin\phi_{\rm t}]^{\rm T}$ and $\mathbf e_{\rm u}=[\cos\phi_{\rm u},\sin\phi_{\rm u}]^{\rm T}$. Define $\delta_{\rm B}\triangleq\phi_{\rm t}-\phi_{\rm u}$. Then,
\begin{equation}
    h\approx 1-\cos\delta_{\rm B},\quad
    S_{xy}^2\approx
    \frac{1}{4}r_{\rm u}^2r_{\rm t}^2\sin^2\delta_{\rm B}.
\end{equation}
Substituting these approximations into the geometric two-dimensional CRLB yields
\begin{equation}
\begin{aligned}
    \mathcal C_{\rm 2D}^{\rm BS}
    &\approx
    \frac{c^2r_{\rm u}^2r_{\rm t}^2}
    {\xi_\tau(1-\cos\delta_{\rm B})^2}
    +
    \frac{r_{\rm u}^2r_{\rm t}^4}
    {\xi_y\cos^2\phi_{\rm t}}
    \left[
    1+
    \frac{\sin^2\delta_{\rm B}}
    {(1-\cos\delta_{\rm B})^2}
    \right].
\end{aligned}
\end{equation}
Using $1+\sin^2 x/(1-\cos x)^2=2/(1-\cos x)$, we obtain
\begin{equation}
    \mathcal C_{\rm 2D}^{\rm BS}
    \approx
    A_{\rm BS}(\phi_{\rm t})r_{\rm t}^2
    +
    B_{\rm BS}(\phi_{\rm t})r_{\rm t}^4,
\end{equation}
where $A_{\rm BS}(\phi_{\rm t})$ and $B_{\rm BS}(\phi_{\rm t})$ are defined in Proposition~\ref{prop:2d_bs_near_coverage}. Solving $A_{\rm BS}r_{\rm t}^2+B_{\rm BS}r_{\rm t}^4=\Gamma$ with respect to $r_{\rm t}^2$ gives the boundary in Proposition~\ref{prop:2d_bs_near_coverage}.

For the 3D case, let $\chi_{\rm B}$ denote the angle between $\mathbf e_{\rm t}$ and $\mathbf e_{\rm u}$, i.e., $\cos\chi_{\rm B}=\mathbf e_{\rm t}^{\rm T}\mathbf e_{\rm u}$. Then $h\approx 1-\cos\chi_{\rm B}$. Moreover,
\begin{equation}
    S_{xy}^2\approx
    \frac{1}{4}r_{\rm u}^2r_{\rm t}^2(\Delta_{xy}^{\rm B})^2,\quad
    S_{xz}^2\approx
    \frac{1}{4}r_{\rm u}^2r_{\rm t}^2(\Delta_{xz}^{\rm B})^2,
\end{equation}
where $\Delta_{xy}^{\rm B}$ and $\Delta_{xz}^{\rm B}$ are defined before Proposition~\ref{prop:3d_bs_near_coverage}. Substituting the above approximations into the geometric 3D CRLB gives
\begin{equation}
\begin{aligned}
    \mathcal C^{\rm BS}
    &\approx
    \frac{c^2r_{\rm u}^2r_{\rm t}^2}
    {\xi_\tau(1-\cos\chi_{\rm B})^2}
    +
    \frac{r_{\rm u}^2r_{\rm t}^4}{\ell_x^2}
    \left[
    \frac{1-\ell_z^2}{\xi_y}
    +
    \frac{1-\ell_y^2}{\xi_z}
    \right.
    \\
    &\qquad\left.
    +
    \frac{1}{(1-\cos\chi_{\rm B})^2}
    \left(
    \frac{(\Delta_{xy}^{\rm B})^2}{\xi_y}
    +
    \frac{(\Delta_{xz}^{\rm B})^2}{\xi_z}
    \right)
    \right].
\end{aligned}
\end{equation}
Collecting the coefficients of $r_{\rm t}^2$ and $r_{\rm t}^4$ yields
\begin{equation}
    \mathcal C^{\rm BS}
    \approx
    A_{\rm B}(\theta_{\rm t},\phi_{\rm t})r_{\rm t}^2
    +
    B_{\rm B}(\theta_{\rm t},\phi_{\rm t})r_{\rm t}^4,
\end{equation}
where $A_{\rm B}(\theta_{\rm t},\phi_{\rm t})$ and $B_{\rm B}(\theta_{\rm t},\phi_{\rm t})$ are defined in Proposition~\ref{prop:3d_bs_near_coverage}. Solving $A_{\rm B}r_{\rm t}^2+B_{\rm B}r_{\rm t}^4=\Gamma$ with respect to $r_{\rm t}^2$ gives the boundary in Proposition~\ref{prop:3d_bs_near_coverage}. This completes the proof.

\section{Proof of UAV-Near Coverage Approximations}
\label{app:uav_near_coverage}

We next consider the UAV-near region. The target location is locally written as $\mathbf q_{\rm t}=\mathbf q_{\rm u}+r_{\rm ut}\mathbf e_{\rm ut}$, where $\mathbf e_{\rm ut}$ is the local UAV-target unit direction. When $r_{\rm ut}\ll r_{\rm u}$, we have $r_{\rm t}\approx r_{\rm u}$ and $\mathbf e_r=\mathbf q_{\rm t}/r_{\rm t}\approx\mathbf e_{\rm u}$. Therefore,
\begin{equation}
    h=1+\mathbf e_r^{\rm T}\mathbf e_{\rm ut}
    \approx
    1+\mathbf e_{\rm u}^{\rm T}\mathbf e_{\rm ut}.
\end{equation}

For the two-dimensional case, let $\mathbf e_{\rm ut}=\mathbf v=[\cos\beta,\sin\beta]^{\rm T}$ and define $\delta_{\rm U}\triangleq\beta-\phi_{\rm u}$. Then,
\begin{equation}
    h\approx 1+\cos\delta_{\rm U},\quad
    \ell_x\approx\cos\phi_{\rm u},\quad
    S_{xy}^2\approx
    \frac{1}{4}r_{\rm u}^2r_{\rm ut}^2\sin^2\delta_{\rm U}.
\end{equation}
Substituting these approximations into the geometric two-dimensional CRLB gives
\begin{equation}
\begin{aligned}
    \mathcal C_{\rm 2D}^{\rm UAV}
    &\approx
    \frac{c^2r_{\rm u}^2r_{\rm ut}^2}
    {\xi_\tau(1+\cos\delta_{\rm U})^2}
    +
    \frac{r_{\rm u}^4r_{\rm ut}^2}
    {\xi_y\cos^2\phi_{\rm u}}
    \left[
    1+
    \frac{\sin^2\delta_{\rm U}}
    {(1+\cos\delta_{\rm U})^2}
    \right].
\end{aligned}
\end{equation}
Using $1+\sin^2 x/(1+\cos x)^2=2/(1+\cos x)$, we obtain the local CRLB expression in Proposition~\ref{prop:2d_uav_near_coverage}. Solving $\mathcal C_{\rm 2D}^{\rm UAV}(r_{\rm ut},\beta)=\Gamma$ with respect to $r_{\rm ut}$ gives the boundary in Proposition~\ref{prop:2d_uav_near_coverage}.

For the 3D case, let $\chi_{\rm U}$ denote the angle between $\mathbf e_{\rm ut}$ and $\mathbf e_{\rm u}$, i.e., $\cos\chi_{\rm U}=\mathbf e_{\rm ut}^{\rm T}\mathbf e_{\rm u}$. Then $h\approx 1+\cos\chi_{\rm U}$. The target direction cosines are locally approximated by those of the UAV direction, i.e.,
\begin{equation}
    \ell_x\approx\ell_{x,{\rm u}},\quad
    \ell_y\approx\ell_{y,{\rm u}},\quad
    \ell_z\approx\ell_{z,{\rm u}}.
\end{equation}
Moreover,
\begin{equation}
    S_{xy}^2\approx
    \frac{1}{4}r_{\rm u}^2r_{\rm ut}^2(\Delta_{xy}^{\rm U})^2,\quad
    S_{xz}^2\approx
    \frac{1}{4}r_{\rm u}^2r_{\rm ut}^2(\Delta_{xz}^{\rm U})^2,
\end{equation}
where $\Delta_{xy}^{\rm U}$ and $\Delta_{xz}^{\rm U}$ are defined before Proposition~\ref{prop:3d_uav_near_coverage}. Substituting the above approximations into the geometric 3D CRLB yields
\begin{equation}
\begin{aligned}
    \mathcal C^{\rm UAV}
    &\approx
    r_{\rm ut}^2
    \left\{
    \frac{c^2r_{\rm u}^2}
    {\xi_\tau(1+\cos\chi_{\rm U})^2}
    +
    \frac{r_{\rm u}^4}{\ell_{x,{\rm u}}^2}
    \left[
    \frac{1-\ell_{z,{\rm u}}^2}{\xi_y}
    +
    \frac{1-\ell_{y,{\rm u}}^2}{\xi_z}
    \right.
    \right.
    \\
    &\qquad\left.\left.
    +
    \frac{1}{(1+\cos\chi_{\rm U})^2}
    \left(
    \frac{(\Delta_{xy}^{\rm U})^2}{\xi_y}
    +
    \frac{(\Delta_{xz}^{\rm U})^2}{\xi_z}
    \right)
    \right]
    \right\}.
\end{aligned}
\end{equation}
Thus,
\begin{equation}
    \mathcal C^{\rm UAV}
    \approx
    G_{\rm U}(\theta_{\rm ut},\phi_{\rm ut})r_{\rm ut}^2,
\end{equation}
where $G_{\rm U}(\theta_{\rm ut},\phi_{\rm ut})$ is defined in Proposition~\ref{prop:3d_uav_near_coverage}. Solving $G_{\rm U}r_{\rm ut}^2=\Gamma$ with respect to $r_{\rm ut}$ gives the boundary in Proposition~\ref{prop:3d_uav_near_coverage}. This completes the proof.

\bibliographystyle{IEEEtran}
\bibliography{IEEEabrv,0reference}

@misc{xuxiaoli,
      title={Hybrid Mono- and Bi-static {OFDM-ISAC} via {BS-UE} Cooperation: Closed-Form {CRLB} and Coverage Analysis}, 
      author={Xiaoli Xu and Yong Zeng},
      year={2026},
      eprint={2601.09057},
      archivePrefix={arXiv},
      primaryClass={cs.IT},
      url={https://arxiv.org/abs/2601.09057}, 
}

@misc{zeng2026capacity,
      title={Capacity Characterization and Formation Optimization for Multi-User {MIMO} Communications with {UAV} Swarm}, 
      author={Yong Zeng},
      year={2026},
      eprint={2605.14298},
      archivePrefix={arXiv},
      primaryClass={cs.IT},
      url={https://arxiv.org/abs/2605.14298}, 
}

@ARTICLE{Zhang_ISAC_tutorial,
  author={Zhang, J. Andrew and Rahman, Md. Lushanur and Wu, Kai and Huang, Xiaojing and Guo, Y. Jay and Chen, Shanzhi and Yuan, Jinhong},
  journal=IEEE_Communications_Surveys_Tutorials, 
  title={Enabling Joint Communication and Radar Sensing in Mobile Networks—A Survey}, 
  year={2022},
  volume={24},
  number={1},
  pages={306-345},
  doi={10.1109/COMST.2021.3122519}}

@ARTICLE{Wei_ISAC_IoT,
  author={Wei, Zhiqing and Qu, Hanyang and Wang, Yuan and Yuan, Xin and Wu, Huici and Du, Ying and Han, Kaifeng and Zhang, Ning and Feng, Zhiyong},
  journal=IEEE_Internet_of_Things_Journal, 
  title={Integrated Sensing and Communication Signals Toward {5G-A} and {6G}: A Survey}, 
  year={2023},
  volume={10},
  number={13},
  pages={11068-11092},
  doi={10.1109/JIOT.2023.3235618}}

@ARTICLE{Fei_UAVisac_commmag,
  author={Fei, Zesong and Wang, Xinyi and Wu, Nan and Huang, Jingxuan and Zhang, J. Andrew},
  journal=IEEE_Communications_Magazine, 
  title={Air-Ground Integrated Sensing and Communications: Opportunities and Challenges}, 
  year={2023},
  volume={61},
  number={5},
  pages={55-61},
  doi={10.1109/MCOM.007.2200459}}

@ARTICLE{Mao_uavisac_wirecom,
  author={Mao, Sun and Yang, Kun and Yuen, Chau},
  journal=IEEE_Wireless_Communications, 
  title={{UAV}-Empowered Integrated Sensing and Communication for {6G}}, 
  year={2026},
  volume={},
  number={},
  pages={1-9},
  doi={10.1109/MWC.2025.3646174}}

@ARTICLE{Gan_ISACcoverage_jsac,
  author={Gan, Xu and Huang, Chongwen and Yang, Zhaohui and Chen, Xiaoming and He, Jiguang and Zhang, Zhaoyang and Yuen, Chau and Liang Guan, Yong and Debbah, Mérouane},
  journal=IEEE_Journal_on_Selected_Areas_in_Communications, 
  title={Coverage and Rate Analysis for Integrated Sensing and Communication Networks}, 
  year={2024},
  volume={42},
  number={9},
  pages={2213-2227},
  doi={10.1109/JSAC.2024.3413989}}

@ARTICLE{jiang_raa_uav_isac,
  author={Jiang, Haoyu and Zeng, Yong},
  journal=IEEE_Transactions_on_Wireless_Communications, 
  title={Ray Antenna Array Achieves Uniform Angular Resolution Cost-Effectively for Low-Altitude {UAV} Swarm {ISAC}}, 
  year={2026},
  volume={25},
  number={},
  pages={9200-9213},
  doi={10.1109/TWC.2025.3643458}}

@ARTICLE{wang_tsp_nearfield_crlb_2024,
  author={Wang, Huizhi and Xiao, Zhiqiang and Zeng, Yong},
  journal=IEEE_Transactions_on_Signal_Processing, 
  title={Cramér-Rao Bounds for Near-Field Sensing With Extremely Large-Scale {MIMO}}, 
  year={2024},
  volume={72},
  number={},
  pages={701-717},
  doi={10.1109/TSP.2024.3350329}}

@ARTICLE{Yuan_nearfield_tccn,
  author={Yuan, Minghao and He, Dongxuan and Yuan, Weijie and Yin, Hao and Wang, Hua},
  journal=IEEE_Transactions_on_Cognitive_Communications_and_Networking, 
  title={Near-Field Integrated Sensing and Communication: {SPEB} Analysis and Hybrid Beamforming Design}, 
  year={2026},
  volume={12},
  number={},
  pages={3511-3524},
  doi={10.1109/TCCN.2025.3602802}}

@ARTICLE{zeng_trajectory_twc,
  author={Zeng, Yong and Zhang, Rui},
  journal=IEEE_Transactions_on_Wireless_Communications, 
  title={Energy-Efficient {UAV} Communication With Trajectory Optimization}, 
  year={2017},
  volume={16},
  number={6},
  pages={3747-3760},
  doi={10.1109/TWC.2017.2688328}}

@book{van1968detection,
  author    = {Harry L. Van Trees},
  title     = {Detection, Estimation, and Modulation Theory, Part I},
  publisher = {John Wiley \& Sons},
  year      = {1968}
}

@ARTICLE{qianglong_tutorial,
  author={Dai, Qianglong and Zeng, Yong and Wang, Huizhi and You, Changsheng and Zhou, Chao and Cheng, Hongqiang and Xu, Xiaoli and Jin, Shi and Lee Swindlehurst, A. and Eldar, Yonina C. and Schober, Robert and Zhang, Rui and You, Xiaohu},
  journal=IEEE_Communications_Surveys_Tutorials, 
  title={A Tutorial on {MIMO-OFDM ISAC}: From Far-Field to Near-Field}, 
  year={2026},
  month={Jan.},
  volume={28},
  number={},
  pages={4319-4358},
  doi={10.1109/COMST.2025.3650568}}

@ARTICLE{yuxuan_magazine,
  author={Song, Yuxuan and Zeng, Yong and Yang, Yuhang and Ren, Zixiang and Cheng, Gaoyuan and Xu, Xiaoli and Xu, Jie and Jin, Shi and Zhang, Rui},
  journal=IEEE_Communications_Magazine, 
  title={An Overview of Cellular {ISAC} for Low-Altitude {UAV}: New Opportunities and Challenges}, 
  year={2025},
  month={Dec.},
  volume={63},
  number={12},
  pages={88-95},
  doi={10.1109/MCOM.002.2400742}}

@ARTICLE{min_sparse_tsp,
  author={Min, Hongqi and Li, Xinrui and Li, Ruoguang and Zeng, Yong},
  journal=IEEE_Transactions_on_Signal_Processing, 
  title={Integrated Localization and Communication With Sparse {MIMO}: Will Virtual Array Technology Also Benefit Wireless Communication?}, 
  year={2025},
  volume={73},
  number={},
  pages={5090-5105},
  doi={10.1109/TSP.2025.3637278}}

@ARTICLE{Xiu_movable_twc,
  author={Xiu, Yue and Zhao, Yang and Yang, Ran and Lyu, Wanting and Niyato, Dusit and In Kim, Dong and Liu, Guangyi and Wei, Ning},
  journal=IEEE_Transactions_on_Wireless_Communications, 
  title={Robust Optimization for Movable Antenna-Aided Cell-Free {ISAC} With Time Synchronization Errors}, 
  year={2026},
  volume={25},
  number={},
  pages={10082-10097},
  doi={10.1109/TWC.2026.3652128}}

@ARTICLE{Jiang_pinching_tcom,
  author={Jiang, Hao and Ouyang, Chongjun and Wang, Zhaolin and Liu, Yuanwei and Nallanathan, Arumugam and Ding, Zhiguo},
  journal=IEEE_Transactions_on_Communications, 
  title={Pinching-Antenna-Assisted Sensing: A Bayesian {Cramér–Rao} Bound Perspective}, 
  year={2026},
  volume={74},
  number={},
  pages={9075-9092},
  doi={10.1109/TCOMM.2026.3690382}}

@ARTICLE{Wu_ckm_twc,
  author={Wu, Di and Dai, Zhuoyin and Zeng, Yong},
  journal=IEEE_Transactions_on_Wireless_Communications, 
  title={You May Use the Same Channel Knowledge Map for Environment-Aware {NLoS} Sensing and Communication}, 
  year={2026},
  volume={25},
  number={},
  pages={14627-14641},
  doi={10.1109/TWC.2026.3678277}}

@ARTICLE{Mu_UAV_comag,
  author={Mu, Junsheng and Zhang, Ronghui and Cui, Yuanhao and Gao, Ning and Jing, Xiaojun},
  journal=IEEE_Communications_Magazine, 
  title={{UAV} Meets Integrated Sensing and Communication: Challenges and Future Directions}, 
  year={2023},
  volume={61},
  number={5},
  pages={62-67},
  doi={10.1109/MCOM.008.2200510}}

@ARTICLE{liu_jsac_isac_2022,
  author={Liu, Fan and Cui, Yuanhao and Masouros, Christos and Xu, Jie and Han, Tony Xiao and Eldar, Yonina C. and Buzzi, Stefano},
  journal=IEEE_Journal_on_Selected_Areas_in_Communications, 
  title={Integrated Sensing and Communications: Toward Dual-Functional Wireless Networks for {6G} and Beyond}, 
  year={2022},
  volume={40},
  number={6},
  pages={1728-1767},
  doi={10.1109/JSAC.2022.3156632}}

@ARTICLE{liu_cst_limits_2022,
  author={Liu, An and Huang, Zhe and Li, Min and Wan, Yubo and Li, Wenrui and Han, Tony Xiao and Liu, Chenchen and Du, Rui and Tan, Danny Kai Pin and Lu, Jianmin and Shen, Yuan and Colone, Fabiola and Chetty, Kevin},
  journal=IEEE_Communications_Surveys_Tutorials, 
  title={A Survey on Fundamental Limits of Integrated Sensing and Communication}, 
  year={2022},
  volume={24},
  number={2},
  pages={994-1034},
  doi={10.1109/COMST.2022.3149272}}

@ARTICLE{zeng_procieee_uav_2019,
  author={Zeng, Yongs and Wu, Qingqing and Zhang, Rui},
  journal=Proceedings_of_the_IEEE, 
  title={Accessing From the Sky: A Tutorial on {UAV} Communications for {5G} and Beyond}, 
  year={2019},
  volume={107},
  number={12},
  pages={2327-2375},
  doi={10.1109/JPROC.2019.2952892}}

@ARTICLE{mozaffari_cst_uav_2019,
  author={Mozaffari, Mohammad and Saad, Walid and Bennis, Mehdi and Nam, Young-Han and Debbah, Mérouane},
  journal=IEEE_Communications_Surveys_Tutorials, 
  title={A Tutorial on {UAVs} for Wireless Networks: Applications, Challenges, and Open Problems}, 
  year={2019},
  volume={21},
  number={3},
  pages={2334-2360},
  doi={10.1109/COMST.2019.2902862}}

@ARTICLE{jiang_commmag_lae_isac_2025,
  author={Jiang, Yihang and Li, Xiaoyang and Zhu, Guangxu and Li, Hang and Deng, Jing and Han, Kaifeng and Shen, Chao and Shi, Qingjiang and Zhang, Rui},
  journal=IEEE_Communications_Magazine, 
  title={Integrated Sensing and Communication for Low Altitude Economy: Opportunities and Challenges}, 
  year={2025},
  volume={63},
  number={12},
  pages={72-78},
  doi={10.1109/MCOM.001.2400685}}

@book{kay_book_1993,
  title={Fundamentals of statistical signal processing: Estimation theory},
  author={Kay, Steven M.},
  volume={1},
  year={1993},
  publisher={Prentice-Hall PTR}
}

@ARTICLE{stoica_nehorai_tassp_1989,
  author={Stoica, P. and Nehorai, Arye},
  journal=IEEE_Transactions_on_Acoustics_Speech_and_Signal_Processing, 
  title={{MUSIC}, maximum likelihood, and {Cramer-Rao} bound}, 
  year={1989},
  volume={37},
  number={5},
  pages={720-741},
  doi={10.1109/29.17564}}

@book{stoica_moses_book_2005,
  title={Spectral Analysis of Signals},
  author={Stoica, Petre and Moses, Randolph L},
  year={2005},
  publisher={Pearson Prentice Hall},
  address={Upper Saddle River, NJ},
  isbn={0131139568}
}

@ARTICLE{liu_tsp_crlb_2022,
  author={Liu, Fan and Liu, Ya-Feng and Li, Ang and Masouros, Christos and Eldar, Yonina C.},
  journal=IEEE_Transactions_on_Signal_Processing, 
  title={{Cramér-Rao} Bound Optimization for Joint Radar-Communication Beamforming}, 
  year={2022},
  volume={70},
  number={},
  pages={240-253},
  doi={10.1109/TSP.2021.3135692}}

@ARTICLE{hua_twc_mimo_2024,
  author={Hua, Haocheng and Han, Tony Xiao and Xu, Jie},
  journal=IEEE_Transactions_on_Wireless_Communications, 
  title={{MIMO} Integrated Sensing and Communication: {CRB}-Rate Tradeoff}, 
  year={2024},
  volume={23},
  number={4},
  pages={2839-2854},
  doi={10.1109/TWC.2023.3303326}}

@ARTICLE{ren_twc_fundamental_2024,
  author={Ren, Zixiang and Peng, Yunfei and Song, Xianxin and Fang, Yuan and Qiu, Ling and Liu, Liang and Ng, Derrick Wing Kwan and Xu, Jie},
  journal=IEEE_Transactions_on_Wireless_Communications, 
  title={Fundamental {CRB}-Rate Tradeoff in Multi-Antenna {ISAC} Systems With Information Multicasting and Multi-Target Sensing}, 
  year={2024},
  volume={23},
  number={4},
  pages={3870-3885},
  doi={10.1109/TWC.2023.3312723}}

@ARTICLE{song_tsp_irs_2023,
  author={Song, Xianxin and Xu, Jie and Liu, Fan and Han, Tony Xiao and Eldar, Yonina C.},
  journal=IEEE_Transactions_on_Signal_Processing, 
  title={Intelligent Reflecting Surface Enabled Sensing: {Cramér-Rao} Bound Optimization}, 
  year={2023},
  volume={71},
  number={},
  pages={2011-2026},
  doi={10.1109/TSP.2023.3280715}}

@ARTICLE{cheng_twc_networked_2024,
  author={Cheng, Gaoyuan and Fang, Yuan and Xu, Jie and Ng, Derrick Wing Kwan},
  journal=IEEE_Transactions_on_Wireless_Communications, 
  title={Optimal Coordinated Transmit Beamforming for Networked Integrated Sensing and Communications}, 
  year={2024},
  volume={23},
  number={8},
  pages={8200-8214},
  doi={10.1109/TWC.2023.3346457}}

@ARTICLE{mao_twc_cellfree_2024,
  author={Mao, Weihao and Lu, Yang and Chi, Chong-Yung and Ai, Bo and Zhong, Zhangdui and Ding, Zhiguo},
  journal=IEEE_Transactions_on_Wireless_Communications, 
  title={Communication-Sensing Region for Cell-Free Massive {MIMO} {ISAC} Systems}, 
  year={2024},
  volume={23},
  number={9},
  pages={12396-12411},
  doi={10.1109/TWC.2024.3392330}}

@ARTICLE{lyu_twc_maneuver_2023,
  author={Lyu, Zhonghao and Zhu, Guangxu and Xu, Jie},
  journal=IEEE_Transactions_on_Wireless_Communications, 
  title={Joint Maneuver and Beamforming Design for {UAV}-Enabled Integrated Sensing and Communication}, 
  year={2023},
  volume={22},
  number={4},
  pages={2424-2440},
  doi={10.1109/TWC.2022.3211533}}

@ARTICLE{deng_twc_adaptable_2023,
  author={Deng, Cailian and Fang, Xuming and Wang, Xianbin},
  journal=IEEE_Transactions_on_Wireless_Communications, 
  title={Beamforming Design and Trajectory Optimization for {UAV}-Empowered Adaptable Integrated Sensing and Communication}, 
  year={2023},
  volume={22},
  number={11},
  pages={8512-8526},
  doi={10.1109/TWC.2023.3264523}}

@STRING{IEEE_Transactions_on_Acoustics_Speech_and_Signal_Processing       = "{IEEE} Trans. Acoust., Speech, Signal Process."}

@STRING{IEEE_Transactions_on_Signal_Processing         = "{IEEE} Trans. Signal Process."}

@STRING{IEEE_Journal_on_Selected_Areas_in_Communications       = "{IEEE} J. Sel. Areas Commun."}

@STRING{IEEE_Transactions_on_Communications        = "{IEEE} Trans. Commun."}

@STRING{IEEE_Transactions_on_Wireless_Communications       = "{IEEE} Trans. Wireless Commun."}

@STRING{IEEE_Internet_of_Things_Journal        = "{IEEE} Internet Things J."}

@STRING{IEEE_Transactions_on_Cognitive_Communications_and_Networking        = "{IEEE} Trans. on Cogn. Commun. Netw."}

@STRING{Proceedings_of_the_IEEE       = "Proc. {IEEE}"}

@STRING{IEEE_Communications_Magazine        = "{IEEE} Commun. Mag."}

@STRING{IEEE_Communications_Surveys_Tutorials      = "{IEEE} Commun. Surveys Tuts."}

@STRING{IEEE_Wireless_Communications         = "{IEEE} Wireless Commun."}

\vfill

\end{document}